\documentclass[a4paper,11pt]{article}
\usepackage[T1]{fontenc}
\usepackage[latin9]{inputenc}
\usepackage{amsmath}
\usepackage{amssymb}
\usepackage{mathtools}
\usepackage{amsthm}

\usepackage{appendix}
\usepackage{hyperref}
\usepackage[dvipsnames]{xcolor}
\definecolor{myrefblue}{RGB}{15,15,150}

\hypersetup{
    colorlinks=true,
    linkcolor=myrefblue,
    citecolor=myrefblue,
    urlcolor=myrefblue,
}

\makeatletter
\@ifundefined{date}{}{\date{}}
\makeatother
\usepackage{fullpage}
\usepackage{graphicx}
\graphicspath{{figures/}}

\usepackage{tikz}
\usepackage{babel}
\usepackage{url}
\usepackage[linesnumbered,ruled,vlined]{algorithm2e}
\usepackage{bbold}
\usepackage{natbib}
\usepackage{thm-restate}
\usepackage{mathrsfs}
\usepackage{enumitem}
\usepackage{xspace}
\usepackage{comment}
\usepackage{subcaption}
\usetikzlibrary{patterns}
\usetikzlibrary{3d}
\usepackage{pgfplots}
\pgfplotsset{compat=1.18}
\usepackage{cleveref}

\newtheorem{mainthm}{Main Theorem}
\AddToHook{env/mainthm/begin}{\crefalias{mainthm}{mainthm}}
\newtheorem*{unnumcorollary}{Corollary}
\newtheorem{proposition}{Proposition}[section]
\AddToHook{env/proposition/begin}{\crefalias{section}{proposition}}
\newtheorem{lemma}{Lemma}[section]
\AddToHook{env/lemma/begin}{\crefalias{section}{lemma}}

\AddToHook{env/claim/begin}{\crefalias{section}{claim}}
\newtheorem{theorem}{Theorem}[section]
\AddToHook{env/theorem/begin}{\crefalias{section}{theorem}}
\newtheorem{corollary}{Corollary}[section]
\AddToHook{env/corollary/begin}{\crefalias{section}{corollary}}
\newtheorem{definition}{Definition}[section]
\AddToHook{env/definition/begin}{\crefalias{section}{definition}}
\newtheorem{remark}{Remark}[section]
\AddToHook{env/remark/begin}{\crefalias{section}{remark}}
\newtheorem{observation}{Observation}[section]
\AddToHook{env/observation/begin}{\crefalias{section}{observation}}

\AddToHook{env/assumption/begin}{\crefalias{section}{assumption}}
\newtheorem{example}{Example}[section]
\AddToHook{env/example/begin}{\crefalias{section}{example}}
\newtheorem{conjecture}{Conjecture}[section]
\AddToHook{env/conjecture/begin}{\crefalias{section}{conjecture}}

\AddToHook{cmd/appendix/after}{%
  \crefalias{section}{appendix}%
  \crefalias{subsection}{appendix}
  \crefalias{subsubsection}{appendix}
}

\begin{document}
\newcommand{\eps}{\varepsilon}
\newcommand{\argmax}{\arg\max}
\newcommand{\argmin}{\arg\min}
\newcommand{\reals}{\mathbb{R}}
\newcommand{\lincontractline}{\mathcal{T}}
\newcommand{\lip}{\mathcal{L}}
\newcommand{\actions}{{[n]}}
\newcommand{\indicator}{\mathbb{1}}
\newcommand{\gooddemand}{ASC\xspace}
\newcommand{\positivenats}{\mathbb{N}_+}
\newcommand{\standard}[1]{\mathbf{e}^{#1}}

\newcommand{\demandtype}{C_{type}}
\newcommand{\demandcover}{V}
\newcommand{\demandcoverfam}{\mathcal{\demandcover}}
\newcommand{\mincover}{\demandcover_{min}}
\newcommand{\demandtypefam}{\mathcal{C}_{type}}
\newcommand{\mincoverfam}{\demandcoverfam_{min}}
\newcommand{\ray}[1]{R(#1)}
\newcommand{\demandset}[2]{{D_{#1}(#2)}}

\newcommand{\neighbors}[2]{A_{#1}(#2)}
\SetKw{KwBreak}{break}

\definecolor{color12}{RGB}{27, 158, 119}
\definecolor{color2}{RGB}{117, 112, 179}
\definecolor{color1}{RGB}{217, 95, 2}

\title{Combinatorial Contracts Through Demand Types\thanks{
We would like to thank Paul Klemperer and Liat Yashin for valuable discussions.
This project has been partially funded by the European Research Council (ERC) under the European Union's Horizon Europe Program (grant agreement No.~101170373 and No.~949699), by an Amazon Research Award, by the NSF-BSF (grant number 2020788), and by the Israel Science Foundation Breakthrough Program (grant No.~2600/24).
}} 
\author{Elizabeth Baldwin\thanks{University of Oxford, United Kingdom. Email: \texttt{elizabeth.baldwin@economics.ox.ac.uk}} \and Paul D\"utting\thanks{Google Research, Zurich, Switzerland. Email: \texttt{duetting@google.com}} \and Michal Feldman\thanks{Tel Aviv University, Israel. Email: \texttt{mfeldman@tauex.tau.ac.il}} \and Maya Schlesinger\thanks{Tel Aviv University, Israel. Email: \texttt{mayas1@mail.tau.ac.il}} }

\maketitle

\begin{abstract}
In the combinatorial action model of contract design, a principal delegates a complex project to an agent, incentivizing a subset of actions from a ground set of $n$ actions, via a linear contract. Computing the optimal contract is a challenging problem that generally hinges on two factors: (i) the number of ``critical values''---values of the linear contract parameter at which the agent's best response changes from one set to another, and (ii)~the complexity of the agent's best-response problem (demand query). 
Prior work has used this approach to devise polynomial-time algorithms for the optimal contract problem under specific reward functions: gross substitutes, supermodular, and ultra.

We develop a unified geometric framework for algorithmic contract design by establishing a fundamental link to the 
theory of \emph{demand types} from consumer theory.
Under this geometric view, bounding the number of critical values reduces to counting the best-response 
regions which the ``contract ray'' pierces.
Leveraging this connection, we introduce the class of \emph{All Substitutes and Complements} (ASC) functions, 
and show that it admits at most $O(n^2)$ critical values, strictly generalizing and unifying all previously known classes admitting poly-many critical values. 
We conjecture that, under some mild assumptions, ASC is the maximal such class.
Turning to the demand query aspect, we develop a new technique for efficiently computing a demand query using value queries, which works in general for ``succinct'' demand types. 
Combining these structural and algorithmic results, we obtain polynomial-time algorithms for new classes of reward functions that exhibit substitutes and complements simultaneously.
\end{abstract}

\section{Introduction}

Algorithmic contract design is a vibrant new frontier at the intersection of Economics and Computation \citep{DuettingFT24,Feldman26}. A central model in this landscape is the \emph{combinatorial action model} \citep{DuettingEFK21,DEFK21-journal,DuettingFG24,VuongDPP24,DuettingFGR26}. In this model, a single principal delegates the execution of a complex project to an agent. The agent takes a set of actions, which can be any subset from a given set of $n$ available actions. The project outcome is binary---resulting in either success or failure. A set function $f$, known as the reward function, maps the chosen subset of actions to a probability of success. The principal earns a (normalized) reward of   $1$ if the project succeeds, and a reward of zero otherwise. The agent incurs an additive cost $c_i$ for each action $i$ included in the set.

The principal, who cannot directly observe the agent's actions, designs a \emph{contract}---an outcome-contingent payment scheme---to incentivize the agent to take a desired set of actions. An optimal contract is a contract that maximizes the principal's utility, defined as the expected reward minus the expected payments to the agent. In the binary outcome setting, the optimal contract takes a particularly appealing and simple form: it is linear. Such a contract is fully specified by a single parameter $\alpha \in [0,1]$, representing the fraction of the principal's reward that is transferred to the agent in the event of success.

Despite this structural simplicity, computing an optimal linear contract is a challenging computational problem. This difficulty stems from two distinct but interrelated computational questions: (a) the agent's problem of determining a best-response action set given a specific contract $\alpha$, and (b) the principal's problem of designing an optimal $\alpha$ to maximize utility, anticipating the agent's response.

A crucial role in addressing these challenges is played by the \emph{upper envelope} of the agent's utility functions. This envelope traces the agent's maximum utility as a function of the contract parameter $\alpha$. Of particular importance are the transition points (or critical values) where the agent switches between two best-response sets. Specifically, if one has access to an efficient algorithm for answering a \emph{demand query} (solving the agent's problem), and the number of critical points is polynomially bounded, then the optimal contract can be computed using only polynomially many value queries via the Eisner-Severance technique \citep{EisnerSeverance1976,DuettingFG24}.

This approach has been the cornerstone of recent progress, yielding poly-time algorithms for specific classes of functions, including gross substitutes \citep{DuettingEFK21,DEFK21-journal}, supermodular \citep{DuettingFG24,VuongDPP24}, and ultra \citep{ultracontracts}. 
This approach remains the only known technique for efficient contract design in this domain. Each of the aforementioned results required developing specialized techniques to bound the critical points for their respective classes.

Beyond the role played in this blueprint for designing efficient algorithms, the number of critical points is a key structural property in its own right, closely related to the pseudodimension of the principal's optimal utility function for linear contracts over a given contract class \citep{DuettingFTS25}, offering insights that extend beyond specific algorithmic applications.

In this work, we establish a fundamental link between algorithmic contract theory and the demand-type machinery from consumer theory, 
introduced by \citet{demand_types}. We utilize this connection to unify all previous results under a single framework and to push the tractability frontier further, identifying additional classes of functions that admit poly-time algorithms. 

\subsection{Our Contribution}

\paragraph{Critical values via demand types.}

In consumer theory, 
a demand type is a collection of set functions $f:2^{[n]}\to \reals_{\ge 0}$ characterized by the geometry of the agent's best-response correspondence (see \citep{demand_types}, \Cref{sec:demand_types}).
Consider the space of item prices, $\mathbb{R}_{\geq 0}^n$. This space is partitioned into convex regions where the agent's unique demand is constant; these regions are separated by hyperplane sections (or lower-dimensional surfaces) where the demand is non-unique. The geometry of these separating hyperplanes defines the demand type: specifically, the normal vectors to these hyperplanes correspond precisely to the items (or bundles) that are ``switched'' in or out of the demand set as prices cross the boundary. 

As already noted in prior work \citep{DuettingEFK21,DEFK21-journal},  the agent's utility maximization problem under a linear contract, $\max_{S} (\alpha f(S) - \sum_{j \in S} c_j)$, is structurally equivalent to a standard demand query with prices scaled by the inverse of the contract parameter. Specifically, maximizing utility is equivalent to maximizing $f(S) - \sum_{j \in S} p_j$, where the effective price of action $j$ is $p_j = c_j / \alpha$. Consequently, as the principal varies $\alpha$ from $0$ to $1$, the effective price vector traces a ray in $\mathbb{R}_{\geq 0}^n$ originating at infinity and terminating at the cost vector $c$. 

Bounding the number of critical values for the contract problem is thus reduced to a geometric counting problem: bounding the number of best-response regions this ``contract ray'' pierces.
We analyze this counting problem, by leveraging the structure implied by the demand type characterization of the reward function.

The demand type framework can be applied to classic set functions from economic theory such as gross-substitutes functions (GS) and supermodular functions, and in fact provides an alternative characterization of either of these classes of set functions \citep{demand_types}.  
Formally, a demand type is a collection of set functions over a single dimension $n$, and when discussing classes such as GS or supermodular, we naturally extend this definition to demand-type families, which specify a demand type for each dimension.

We introduce a new demand-type family, \emph{All Substitutes and Complements} (ASC) (see Definition~\ref{def:our-set-functions} and Section~\ref{sec:demand-covers}), and prove that for any reward function in this class, the number of critical values is polynomially bounded. Specifically, we show that the contract ray pierces at most $O(n^2)$ demand regions.  For valuations in this class, any pair of goods may be substitutes and any subset of goods may be complementary.

\begin{mainthm}[Theorem~\ref{thm:gooddemand_poly_crit}]\label{mainres} The demand-type family ASC admits $O(n^2)$ many critical values.
\end{mainthm}
The ASC class strictly contains the classes of gross substitutes, supermodular, and ultra valuations, thereby unifying the distinct tractability results 
found in prior work \citep{DuettingEFK21,DEFK21-journal,DuettingFG24,VuongDPP24,ultracontracts}. Additionally, ASC contains the class of \emph{Gross Substitutes and Complements} (GSC) introduced by \citet{sun_and_yang}, and in fact a wider class which we refer to as GSC+ (see Definition~\ref{def:our-set-functions} and \Cref{prop:gsc+}).
The class GSC partitions the goods into two sets; any pair of goods within a set is substitutes while any pair across the two sets is complements. GSC+ generalizes this by allowing any pair to be either substitutes or complements.
See Figure~\ref{fig:asc_classes_venn} for a Venn diagram relating the different classes to each other. Since ASC contains GS, \Cref{mainres} is tight by a result of \citet{DEFK21-journal}; there exist GS functions with $\Omega(n^2)$ critical points.
\definecolor{ultracolor}{RGB}{117,112,179} 
\definecolor{gccolor}   {RGB}{217,95,2}    
\definecolor{gsccolor}  {RGB}{27,158,119}  

\definecolor{gscolor}{RGB}{204,76,153}
\begin{figure}
    \centering
    \begin{tikzpicture}[scale=0.7]
        \tikzset{
            setstyle/.style 2 args={
                draw=#1, fill=#1, thick, fill opacity=0.1, text opacity=1, text=#2
            }
        }

        \draw[setstyle={gccolor}{gccolor}] (0, 0) ellipse (4 and 2);
        \node[text=gccolor, align=center] at (-1.35, 0) {gross complements/ \\
        supermodular};

        \draw[setstyle={gsccolor}{gsccolor}] (5, 0) ellipse (4 and 2);
        \node[text=gsccolor] at (7, 0) {GSC};

        \draw[setstyle={ultracolor}{ultracolor}] (2.5, 2.5) ellipse (4 and 2);
        \node[text=ultracolor] at (2.5, 4) {ultra};

        \draw[draw=gscolor, fill=gscolor, thick, fill opacity=0.1, rotate=10] (4, 0.6) ellipse (1.5 and 0.5);
        \node[text=gscolor] at (4, 1.25) {GS};

        \draw[rounded corners=10pt, thick, black] (-4.5, -2.5) rectangle (9.5, 5);
        \node at (-3, 4.5) {\textbf{ASC}};

    \end{tikzpicture}
    \caption{A Venn diagram of the classes of set functions studied in this paper. Explanations for the containment relations and examples showing that each region is non-empty can be found in \Cref{sec:venn_diagram_examples}.}
    \label{fig:asc_classes_venn}
\end{figure}
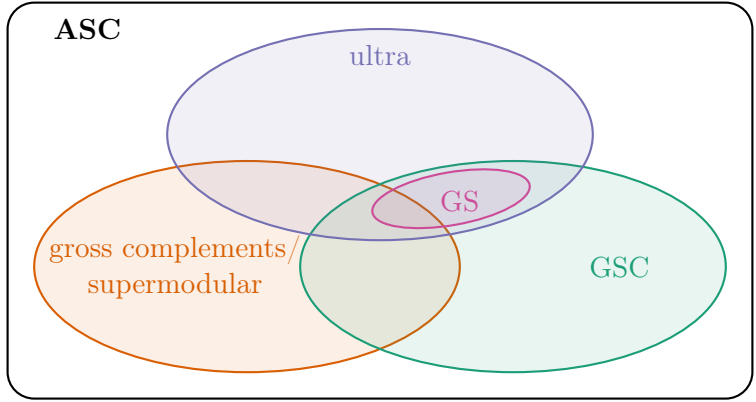

To establish our bound, we define the notion of ``facet-piercing'' (see Definition~\ref{def:facet-piercing}), which generalizes the genericity assumptions required in prior work \citep{DuettingEFK21,DEFK21-journal,ultracontracts}. 
For ASC, the geometric structure imposes strict limits on the possible intersections along the contract ray. 
Specifically, we can use the same potential function of \citet{DuettingEFK21,DEFK21-journal} to bound the number of critical values under the facet-piercing assumption. Finally, a standard perturbation argument allows us to lift the facet-piercing assumption.

\paragraph{Efficient demand queries.}
Our second major contribution is a method for deriving efficient demand queries directly from the demand type structure. We define the notion of a \emph{succinct demand-type family}---a class of valuations where the set of relevant normal vectors is polynomially bounded and known (see \Cref{def:succinct}). We prove that if a class is a succinct demand-type family then a demand query can be answered using only polynomially many value queries (black-box access to $f(S)$). 

\begin{mainthm}\label{main_thm:demand_query} Any succinct demand-type family admits an efficient demand query algorithm.
\end{mainthm}

Notably, this result doesn't require any other economic properties of the demand-type family beyond succinctness. The basic idea of our algorithm is to trace a path through price-space, starting at a point where the empty set is demanded, and ending at our desired price vector. This path is essentially changing the prices coordinate by coordinate, and for each coordinate, the demanded bundle changes at most once according to the demand type characterization; since there are $n$ items, there are $n$ such coordinates to go over. The succinctness of the demand-type family allows us to efficiently maintain the currently demanded bundle as we move along the path, eventually leading to the demanded bundle at our desired prices.

As a corollary, we obtain an efficient demand query algorithm for GSC and for the broader class of GSC+ (Corollary~\ref{cor:gsc+_demand_query}).\footnote{For the more restrictive class of GSC an efficient demand query also follows from a simple basis change argument, but this argument does not extend to GSC+.} 
Another interesting succinct demand-type family is \emph{$\Delta$-substitutes} \citep{near_substitutes}, which allows substitution between complementary 
packages of size up to $\Delta$.  Thus, the class of $\Delta$-substitutes functions admits an efficient demand query.

In combination with the Eisner-Severance technique \citep{EisnerSeverance1976,DuettingFT24}, \Cref{mainres} and \Cref{main_thm:demand_query} yield efficient algorithms for the optimal contract problem for the following classes.

\begin{unnumcorollary}[Corollary \ref{cor:opt_computation}] The classes GSC, GSC+, and the intersection of $\Delta$-substitutes and ASC all admit a poly-time algorithm for computing the optimal contract.
\end{unnumcorollary}

While our work relies on the machinery developed by \citet{demand_types} to study competitive equilibrium existence, equilibrium existence is not the driving force behind our structural results and polynomial-time computation. 
Indeed, both GSC+ and ultra allow for efficient optimization, but sometimes do not admit competitive equilibria. 

We conjecture that ASC is, in a specific sense, maximal. Subject to the axioms of \emph{dimensional-invariance}\footnote{Roughly speaking, dimensional-invariance means that ``any additional dimension may be treated as redundant'', see \Cref{def:consistent_fam}.} and \emph{anonymity}\footnote{Intuitively, anonymity means that ``the names of the items do not matter'', see \Cref{def:anonymous}.} we provide evidence suggesting that no broader class admits a polynomial bound on the number of critical values, see \Cref{app:asc_max}. 

In \Cref{app:limits}, we further show that certain classes, including coverage functions and budget-additive functions,  may only be represented as very strict subsets of very broad demand types. This implies that viewing these classes through the demand types lens loses a lot of their structure.

\paragraph{Broader impact.} 

Our contributions have direct implications regarding the foundations of combinatorial auctions and algorithmic game theory, specifically regarding the power of query oracles. While demand queries can solve the welfare-maximizing Configuration LP in polynomial time \citep{BlumrosenN10}, they are exponentially more powerful than value queries for general valuations. 

We overcome this barrier using a geometric approach; Theorem~\ref{thm:demand_query} provides a generic reduction showing that for valuations with succinct demand-type families, 
demand queries can be efficiently simulated via value queries. This result extends the reach of the Configuration LP to new valuation classes, offering a geometric lens to drive further advances in combinatorial auctions and contract theory.

More generally we believe that the geometric demand-type perspective offers a fascinating lens on computational aspects of combinatorial auctions and contracts, which  could drive further advances in very much the same way it did for equilibrium existence.

\subsection{Further Related Work}

\paragraph{Combinatorial contract design. }  The combinatorial-action contract design model was introduced in \citet{DuettingEFK21}, who gave a poly-time algorithm for the optimal contract when $f$ is GS using the critical values approach. They also showed this problem is NP-hard when $f$ is submodular. \citet{DuettingFG24} and \citet{VuongDPP24} further studied this model and showed that the optimal contract can be computed in polynomial time when $f$ is supermodular, again relying on the critical values approach. \citet{ultracontracts} extended the techniques of \citet{DuettingEFK21} to ultra valuations, introduced by \citet{ultra_valuations}, a strict superclass of gross substitutes. They also extend this to costs which are the sum of a symmetric function and an additive function.

\citet{inapproximability} study the approximability of the optimal contract problem in the model of \citet{DuettingEFK21}, and show that when $f$ is submodular, the problem is not only NP-hard but NP-hard to approximate within any constant, given only value oracle access. \citet{multimulti} presented an FPTAS for any monotone function, which uses both value and demand oracle access. This FPTAS is complemented by the result of \citet{DuettingFGR26}, which shows that even for submodular $f$, with both value and demand oracle access, the problem cannot be solved efficiently.

Other contract settings with multiple actions include scenarios in which the agent takes actions sequentially rather than simultaneously~\citep{hoefer,robustsequential,ezra2026contract}.
Another line of work considers combinatorial contracts along additional dimensions, including multi-agent contracts~\citep{babaioff2006combinatorial,multi,inapproximability,VuongDPP24,budget,aht}, multi-agent multi-action contracts \citep{multimulti,onetoomany,DEFK25mixed}, settings with multiple outcomes \citep{dutting2021complexity}, and ones with multiple principals \citep{alon2024incomplete}. 
Additional multi-agent settings have been studied in \citep{castiglioni2023multi, cacciamani2024multi, alon2025multi, faircontracts}, considering observable individual outcomes, randomized contracts in multiple-outcome settings, multiple projects and fairness considerations.

\paragraph{Demand types and equilibrium existence.}
The important class of gross substitutes (GS) valuations was introduced by \citet{Kelso.Crawford1982}, further developments and alternative representations include \citep{gul1999walrasian, Milgrom.Strulovici2009, Murota2003}.  These texts typically focus on existence and computation of competitive equilibrium; \citet{sun_and_yang} showed with their Gross Substitutes and Complements (GSC) valuations that this property extends beyond the substitutes case.  The notion of demand types was introduced by \citet{demand_types} to provide a general framework for classes of valuations, including both GS and GSC as classes, and characterising exactly which of these classes guarantee competitive equilibria.  \citet{baldwinworkingpaper} provide alternative characterisations of demand types, while \citet{Baldwin.etal2023} extend some of the results to economies with income effects. \citet{near_substitutes} introduce $\Delta$-substitute preferences also in a setting allowing income effects, and show that they correspond to a demand type in the quasilinear case. The demand types framework has been used by \citet{Baldwin.etal2024} to develop a bidding language that allows representation in an auction of any preferences for substitutes.

\section{Model and Preliminaries}\label{sec:model}
In this section we present the combinatorial contracts model and preliminaries of \citet{DuettingEFK21,DEFK21-journal} (\Cref{sec:contracts_model}), and the commonly used ``critical values'' approach to computing an optimal contract in it , as used by \citet{DuettingEFK21,DEFK21-journal, DuettingFT24,VuongDPP24,ultracontracts} (\Cref{sec:crit_values}).

\subsection{The Combinatorial Contracts Model}\label{sec:contracts_model}
\paragraph{Contracts with combinatorial actions.} 
In the combinatorial contracts model of \citet{DuettingEFK21}, a principal (she) interacts with a single agent (he), in an effort to make a project succeed. The agent faces a set $[n]=\{1,\dots, n\}$ of $n$ actions, and can choose any subset $S \subseteq [n]$ of them. There is a vector of costs $c\in [0,\infty)^n$, where $c_i$ is the cost of action $i\in [n]$.
When choosing a subset of actions $S\subseteq [n]$, the agent incurs a cost of $\sum_{i\in S} c_i$.

The project can either succeed or fail. There is a function $f: 2^{[n]} \rightarrow [0,1]$, where $f(S)$ denotes the probability that the project succeeds when the agent chooses the set of actions $S$. If the project succeeds, the principal gets a reward which we normalize to $1$. Otherwise, her reward is zero. The principal is not aware of which actions the agent has chosen, and can only observe whether the project has succeeded or failed.

To incentivize the agent to exert effort, the principal designs a contract, which defines payments for the observable outcomes. In a binary-outcome setting such as the one considered here, it is known that without loss of generality this contract is linear, meaning that it is specified by a single parameter $\alpha \in [0,1]$ denoting the fraction of the principal's reward that is transferred to the agent when the project succeeds. 

Given a contract $\alpha$, when the agent chooses the subset $S\subseteq [n]$, his utility is defined as the expected payment from the principal minus his cost, i.e., $u_A(\alpha, S) = \alpha \cdot f(S) - \sum_{i\in S} c_i$.
The principal's utility is 
the expected reward minus the expected payment to the agent, i.e.,
$u_P(\alpha, S) = (1-\alpha) f(S)$.

The agent's best response for a contract $\alpha$ is $S_{\alpha} \in \arg\max_{S \subseteq [n]} u_A(\alpha, S)$, where the agent breaks ties in favor of the principal's utility. In the edge case where $\alpha = 1$, we assume tie-breaking in favor of larger $f(S_\alpha)$ for consistency\footnote{When $\alpha=1$, the principal's utility is always $0$. Tie-breaking in favor of larger $f(S_\alpha)$ in this case guarantees that $f(S_\alpha)$ is increasing in $\alpha$.}. 
The principal's goal is to find an optimal contract 
$\alpha^\star$ that maximizes her utility given that the agent best responds, i.e. $\alpha^\star$ maximises
$u_P(\alpha) = u_P(\alpha, S_\alpha)$.

\paragraph{Best response and demand set.} 
The agent's best-response problem of choosing a set $S$ that maximizes $\alpha \cdot f(S) - \sum_{i \in S} c_i$ is closely related to the familiar notion of a demand set in combinatorial auctions.
Namely, given a set function $f: 2^{[n]} \rightarrow [0,1]$ 
and a vector of prices $p \in \reals^n$, 
the 
\emph{demand set} $\demandset{f}{p}$ is defined as 
\[\demandset{f}{p} = \arg\max_{S\subseteq [n]} \left(f(S) - \sum_{i\in S} p_i\right).
\]
That is, 
$\demandset{f}{p}$ is the collection of bundles $S\subseteq [n]$
that
maximize the utility of a buyer in a combinatorial auction with items $[n]$, given by $u_B(p, S) = f(S)-\sum_{i\in S} p_i$, i.e., the difference between the value of the bundle under $f$ and the sum of prices of the bundle's elements according to $p$. 

The connection to the contracts model is now clear:
the agent's best response $S_\alpha$ to a contract $\alpha> 0$ is a set in the demand set $\demandset{\alpha \cdot f}{c}  = \demandset{f}{\left( \frac{c_1}{\alpha},\dots, \frac{c_n}{\alpha}\right)}=\demandset{f}{\alpha^{-1}\cdot c}$. That is, we can equivalently view the agent's best response problem as choosing a set that maximizes the utility of a buyer with valuation function $\alpha \cdot f$ at prices $(c_1,\dots, c_n)$, or equivalently, choosing such a set when the valuation function is $f$ and prices are $(c_1/\alpha, \ldots, c_n/\alpha)$.
\footnote{Note that $(\frac{c_1}{\alpha},\ldots,\frac{c_n}{\alpha})$ is a ray
(unless $(c_1,\ldots,c_n)=0$).  There are finitely many different bundles demanded on this line, and the bundle demanded furthest from $0$ (as we increase all coordinates) also corresponds to the tie-breaking demand when $\alpha=0$.} 

This connection is illustrated by \Cref{fig:agent_br_demand} for a two-action example. On the left, in \Cref{fig:3d_demand}, we see the buyer's utility as a function of prices, and on the right, in \Cref{fig:3d_slice}, we see a specific slice of this which corresponds to the buyer's utility at prices $\alpha^{-1} c$. The agent's best response is the set demanded at these prices.

\newcommand{\figthreedpx}{0.7}
\newcommand{\figthreedpy}{0.7*5/3}
\newcommand{\figthreedfx}{2}
\newcommand{\figthreedfy}{3}
\newcommand{\figthreedfxy}{4}
\begin{figure}
    \centering
    \begin{subfigure}[c]{0.48\textwidth}
        \centering
\newcommand{\drawmax}{4}
\pgfmathsetmacro{\betaonetwo}{(\figthreedfx - \figthreedfy) / (\figthreedpx - \figthreedpy)}
\pgfmathsetmacro{\betaonenull}{(\figthreedfx) / (\figthreedpx)}

\begin{tikzpicture}[scale=0.85]
    \begin{axis}[
        view={120}{30}, 
        axis lines=center,
        xlabel={$p_1$},
        ylabel={$p_2$},
        zlabel={$u_B(p)$},
        xmin=0, xmax=\drawmax,
        ymin=0, ymax=\drawmax,
        zmin=0, zmax=\drawmax,
        ticks=none,
        set layers,
        enlargelimits={upper=0.15},
        xlabel style={at={(ticklabel* cs:1)}, anchor=north west},
        ylabel style={at={(ticklabel* cs:1)}, anchor=south west},
        zlabel style={at={(ticklabel* cs:1)}, anchor=south},
    ]

    \draw[thick, fill=color12!20] (0,0,0) -- ({\figthreedfxy - \figthreedfy},0,0) -- (\figthreedfxy - \figthreedfy,\figthreedfxy - \figthreedfx,0) -- (0,\figthreedfxy - \figthreedfx,0) -- cycle;
    \draw[thick, fill=color1!20] (0,\figthreedfxy - \figthreedfx,0) -- (\figthreedfxy - \figthreedfy,\figthreedfxy - \figthreedfx,0) -- (\figthreedfx,\figthreedfy,0) -- (\figthreedfx,\drawmax,0) -- (0, \drawmax, 0) -- cycle;
    \draw[thick, fill=color2!20] (\figthreedfxy - \figthreedfy,0,0) -- (\figthreedfxy - \figthreedfy,\figthreedfxy - \figthreedfx,0) -- (\figthreedfx,\figthreedfy,0) -- (\drawmax,\figthreedfy,0) -- (\drawmax, 0, 0) -- cycle;
    \begin{scope}[canvas is xy plane at z=0]
        \node[transform shape, rotate=90, font=\Large] at (3, 3.5) {$\emptyset$};
        \node[transform shape, rotate=90, font=\Large] at (0.5, 1) {$\{1,2\}$};
        \node[transform shape, rotate=90, font=\Large] at (3, 1) {$\{2\}$};
        \node[transform shape, rotate=90, font=\Large] at (0.5, 3.5) {$\{1\}$};
    \end{scope}

    \draw[ultra thick] (0,0,0) -- (0,0,\figthreedfxy);
    \draw[ultra thick] (1,2,0) -- (1,2,\figthreedfx-1); 
    \draw[ultra thick] (1,0,0) -- (1,0,\figthreedfy);
    \draw[ultra thick] (0,2,0) -- (0,2,\figthreedfx);
    \draw[ultra thick] (0,\drawmax,0) -- (0,\drawmax,\figthreedfx);
    \draw[ultra thick] (\drawmax,0,0) -- (\drawmax,0,\figthreedfy);
    
    \filldraw[draw=black, fill=color1!30, opacity=0.5] 
        (0,\figthreedfxy-\figthreedfx,\figthreedfx) -- ({\figthreedfxy-\figthreedfy}, \figthreedfxy-\figthreedfx, \figthreedfx-\figthreedfxy+\figthreedfy) --(\figthreedfx,\figthreedfy,0) -- (\figthreedfx,\drawmax,0) -- (0,\drawmax,\figthreedfx) -- cycle;
        
    \filldraw[draw=black, fill=color2!30, opacity=0.5] 
        (\figthreedfxy-\figthreedfy,0,\figthreedfy) -- (\figthreedfxy-\figthreedfy, \figthreedfxy-\figthreedfx, \figthreedfy-\figthreedfxy+\figthreedfx) --(\figthreedfx,\figthreedfy,0) -- (\drawmax,\figthreedfy,0) -- (\drawmax,0,\figthreedfy) -- cycle;
        
    \filldraw[draw=black, fill=color12!30, opacity=0.5] 
        (0,0,\figthreedfxy) -- (\figthreedfxy-\figthreedfy, 0, \figthreedfy) --(\figthreedfxy-\figthreedfy,\figthreedfxy-\figthreedfx,\figthreedfy+\figthreedfx-\figthreedfxy) -- (0,\figthreedfxy-\figthreedfx,\figthreedfx) -- cycle;

    \draw[thick,dashed, -stealth]  (\figthreedpx,\figthreedpy,0)--(\figthreedpx*3.25, \figthreedpy*3.25, 0) node[anchor= west] {$\alpha^{-1}$};
    \filldraw[gray, opacity=0.2, draw=white, thick] 
    (\figthreedpx,\figthreedpy, 0) -- (\figthreedpx*3.2, \figthreedpy*3.2, 0) -- (\figthreedpx*3.2, \figthreedpy*3.2, 2.5) -- (\figthreedpx,\figthreedpy, 2.5) -- cycle;

    \draw[ultra thick, color12] 
        (\figthreedpx, \figthreedpy, \figthreedfxy - \figthreedpx - \figthreedpy) -- (\figthreedfxy-\figthreedfy, \figthreedpx/\figthreedpy, \figthreedfy - \figthreedpx/\figthreedpy); 

    \draw[ultra thick, color2] 
        (\figthreedfxy-\figthreedfy, \figthreedpx/\figthreedpy, \figthreedfy - \figthreedpx/\figthreedpy) -- (\betaonetwo*\figthreedpx, \betaonetwo*\figthreedpy, \figthreedfy - \betaonetwo*\figthreedpy);

    \draw[ultra thick, color1] 
        (\betaonetwo*\figthreedpx, \betaonetwo*\figthreedpy, \figthreedfy - \betaonetwo*\figthreedpy) -- (\figthreedfxy-\figthreedfx, \betaonenull*\figthreedpy, \figthreedfx - \figthreedfxy+\figthreedfx);

    \draw[gray] (\figthreedfxy-\figthreedfy, \figthreedpx/\figthreedpy,0) -- (\figthreedfxy-\figthreedfy, \figthreedpx/\figthreedpy,2.5);

        \draw[gray] (\betaonetwo*\figthreedpx,\betaonetwo*\figthreedpy,0) -- (\betaonetwo*\figthreedpx,\betaonetwo*\figthreedpy,2.5);
        \draw[gray] (\betaonenull*\figthreedpx,\betaonenull*\figthreedpy,0) -- (\betaonenull*\figthreedpx,\betaonenull*\figthreedpy,2.5);
    \end{axis}
\end{tikzpicture}
    \caption{
    The ``roof'' of the buyer's utility, as a function of prices. The $p_1$ and $p_2$ axes are the prices, and the $u_B(p)$ axis is the buyer's utility after purchasing the best bundle at those prices. Each of the colored hyperplanes corresponds to the utility from a different bundle, and they are defined over the area where that bundle is demanded (as visualized at $u_B(p)=0$).
    The dashed arrow represents the ray $(c_1/\alpha, \dots, c_n/\alpha)$, and the gray hyperplane through it is the ``slice'' which generates \Cref{fig:3d_slice}.
    }\label{fig:3d_demand}
    \end{subfigure}
    \hfill
    \begin{subfigure}[c]{0.48\textwidth}
        \centering 
\begin{tikzpicture}[scale=0.85]
    \begin{axis}[
        axis x line=bottom,
        axis y line=left,
        axis line style = {-stealth},   
        x axis line style={dash pattern=on 5pt off 5pt, thick},
        xlabel style={at={(axis description cs:1,0)}, anchor=west},
        ylabel style={at={(axis description cs:0,1)}, anchor=south, rotate=-90},
        xlabel = {$\alpha^{-1}$},
        ylabel = {$u_B(\alpha^{-1} c)$},
        xtick={1,
            (\figthreedfxy - \figthreedfy)/\figthreedpx, 
            (\figthreedfy - \figthreedfx)/(\figthreedpy - \figthreedpx), 
            \figthreedfx/\figthreedpx
        },
        ytick=\empty,
        ymin=0,
        xmin=1,
        xmax=3.2, 
        ymax=2.5,
        xmajorgrids,
        legend pos=north east,
        axis background/.style={fill=gray!20},
        x tick label style={font=\footnotesize},
    ]

    \addplot [
        domain=1:(\figthreedfxy - \figthreedfy)/\figthreedpx,
        ultra thick,
        samples=100, 
        color=color12,
    ]
    {\figthreedfxy - x*(\figthreedpx + \figthreedpy)};
    \addlegendentry{$u_B(\alpha^{-1} c, \{1,2\})$}

    \addplot [
        domain=(\figthreedfxy - \figthreedfy)/\figthreedpx:(\figthreedfy - \figthreedfx)/(\figthreedpy - \figthreedpx), 
        ultra thick,
        samples=100, 
        color=color2,
    ]
    {\figthreedfy - x*\figthreedpy};
    \addlegendentry{$u_B(\alpha^{-1} c, \{2\})$}

    \addplot [
        domain=(\figthreedfy - \figthreedfx)/(\figthreedpy - \figthreedpx):\figthreedfx/\figthreedpx, 
        ultra thick,
        samples=100, 
        color=color1,
    ]
    {\figthreedfx - x*\figthreedpx};
    \addlegendentry{$u_B(\alpha^{-1} c, \{1\})$}

    \path (1, 0) -- node[midway, above, color12] {$\{1,2\}$} ({(\figthreedfxy - \figthreedfy)/\figthreedpx}, 0);
    
    \path ({(\figthreedfxy - \figthreedfy)/\figthreedpx}, 0) -- node[midway, above, color2] {$\{2\}$} ({(\figthreedfy - \figthreedfx)/(\figthreedpy - \figthreedpx)}, 0);
    
    \path ({(\figthreedfy - \figthreedfx)/(\figthreedpy - \figthreedpx)}, 0) -- node[midway, above,color1] {$\{1\}$} ({\figthreedfx/\figthreedpx}, 0);
    
    \path ({\figthreedfx/\figthreedpx}, 0) -- node[midway, above] {$\emptyset$} (3.2, 0);

    \end{axis}
\end{tikzpicture}
        \caption{The buyer's utility at prices $\alpha^{-1}\cdot c$ as a function of $\alpha^{-1}$, and the corresponding demanded bundle. This corresponds to the agent's best response to a contract $\alpha$. The demanded bundles at each interval are written above the $\alpha^{-1}$ axis. 
    }\label{fig:3d_slice}
    \end{subfigure}
    \caption{The relationship between an auction buyer's utility and the agent's best response, for $f(\{1\})=2, f(\{2\})=3, f(\{1,2\})=4$ and $c_1=0.7, c_2=7/6$. 
    }\label{fig:agent_br_demand}
\end{figure}
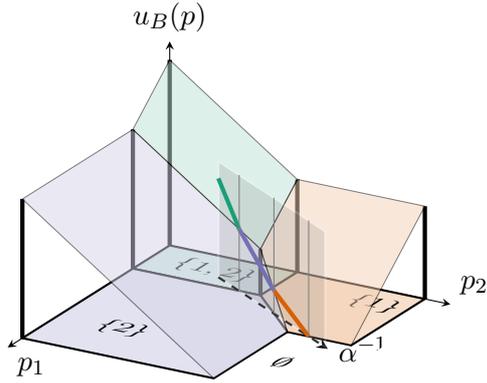
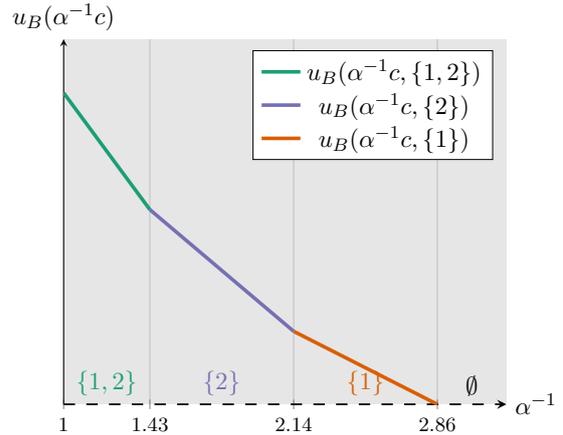

\paragraph{Classes of set functions.}
In this work, we focus on monotone set functions, which are functions $f:2^{[n]}\rightarrow \reals_{\ge 0}$ such that $f(S) \le f(T)$ for all $S\subseteq T\subseteq [n]$. We study the classes of set functions defined below, the relationships between which are summarized in \Cref{fig:asc_classes_venn}. 

Given a set function $f: 2^{[n]} \rightarrow [0,1]$, 
we let $f(j \mid S) = f(S \cup \{j\}) - f(S)$. 
    
\begin{definition}[Gross substitutes, gross complements, and supermodularity]\label{def:common-set-functions}
A set function $f:2^{[n]} \rightarrow \reals_{\ge 0}$ is
\begin{enumerate}
    \item \emph{Gross-substitutes} (GS) if for any price vector $p\in \mathbb{R}^n_{\geq 0}$, any item $i\in [n]$, any $\delta \ge 0$, and any set $S\in \demandset{f}{p}$, there exists a set $T \in \demandset{f}{p+\delta \standard{i}}$ such that $S\setminus \{i\} \subseteq T$, where $\standard{i}$ is the standard basis vector corresponding to $i\in [n]$. \label{def:gs}
    \item \emph{Gross-complements} if for any price vector $p\in \mathbb{R}^n_{\geq 0}$, any item $i\in [n]$, any $\delta \ge 0$, and any set $S\in \demandset{f}{p}$, there exists a set $T \in \demandset{f}{p+\delta \standard{i}}$ such that $T \subseteq S $. \label{def:complements}
    \item \emph{Supermodular} if for any two sets $S \subseteq T \subseteq [n]$ and $j \not \in T$ it holds that $f(j \mid S) \le f(j \mid T)$. \label{def:supermod}   
\end{enumerate}
\end{definition}
Two remarks are in order. First, supermodular and gross complements are equivalent (as \citet{Chambers.Echenique2009} note, this has been known for many years). 
Second, definitions of gross substitutes and gross complements often allow for changes in the prices of more than one item. The definitions where only one item's price, as stated here, are equivalent, and we use them for consistency with the following definitions.

\begin{definition} [GSC, $\Delta$-substitutes, and ultra]\label{def:uncommon-set-functions}
A set function $f:2^{[n]} \rightarrow \reals_{\ge 0}$ is
\begin{enumerate}
    \item  
    \emph{Gross substitutes and complements} (GSC) \citep{sun_and_yang} if there exists a partition $A_1 \uplus A_2 = [n]$ such that (informally)
    items are GS within each part of the partition, and they are complements across parts of the partition.

    Formally, for any price vector $p\in \mathbb{R}^n_{\geq 0}$ any $i\in \{1,2\}$, any item $a\in A_i$, any $\delta \ge 0$ and any set $S\in \demandset{f}{p}$, there exists a set $T\in \demandset{f}{p+\delta \standard{a}}$ such that 
    for any $j\in S\cap A_i\setminus \{a\}$ it holds that $j\in T$ (substitutes within $A_i$), 
    and for any $j\in T \cap A_i^c$ it holds that $j\in S$ (complements between $A_i$ and $A_i^c$). \label{def:gsc}
    
    \item \emph{$\Delta$-substitutes}\footnote{\citet{near_substitutes} give their definition in a slightly different form, which is more general than the quasilinear setting of this paper.  In the quasilinear, as may be seen through the demand type analysis (see \Cref{prop:near-substitutes_demand_type}), the definition given here is equivalent to theirs.} \citep{near_substitutes}, for $\Delta \in \mathbb{N}^+$, if for any price vector $p\in \reals^n$, any item $i\in [n]$, any $\delta \ge 0$, and any set $S\in \demandset{f}{p}$, there exists a set $T\in \demandset{f}{p+\delta \standard{i}}$ such that $|S\setminus T| \le \Delta$ and $|T\setminus S| \le \Delta$.
    \label{def:delta_sub}
    \item \emph{Ultra} \citep{ultra_valuations} if for any two sets $S,T\subseteq [n]$ such that $|S| \le |T|$, and any $x\in S\setminus T$ there exists some $y\in T\setminus S$ such that $f(S)+f(T)\le f((S\setminus \{x\}) \cup \{y\})+f((T\setminus \{y\}) \cup \{x\})$. \label{def:ultra} 
\end{enumerate}
\end{definition}
We remark that $1$-substitutes $= GS\subseteq GSC\subseteq 2$-substitutes, and $GS=$submodular $\cap$ ultra, while submodular and ultra are incomparable.
We also define the two following classes, the motivation for which will become clearer when we discuss demand types in \Cref{sec:demand_types}.
\begin{definition} [GSC+ and ASC] \label{def:our-set-functions}
A set function $f:2^{[n]} \rightarrow \reals_{\ge 0}$ is:
\begin{enumerate}
    \item GSC+  if for any price vector $p\in \mathbb{R}^n_{\geq 0}$, any item $i\in [n]$, any $\delta \ge 0$, and any set $S\in \demandset{f}{p}$, there exists a set $T \in \demandset{f}{p+\delta \standard{i}}$ such that $|(S\setminus \{i\}) \triangle T| \le 1$, where $X\triangle Y$ denotes the symmetric difference between the sets $X,Y$, i.e., $X\triangle Y = (X\setminus Y) \cup (Y\setminus X)$. 
    \item All substitutes and complements (\gooddemand) if for any price vector $p\in \mathbb{R}^n_{\geq 0}$, any item $i\in [n]$, any $\delta \ge 0$, and any set $S\in \demandset{f}{p}$, it holds that there exists a set $T \in \demandset{f}{p+\delta \standard{i}}$ such that either $|(S\setminus \{i\}) \triangle T| \le 1$ or $T\subseteq S$. \label{def:ASC}
\end{enumerate}
\end{definition}
We remark that GSC+ contains GSC, this is not immediate, but follows from the demand type analysis in \Cref{sec:demand_types}. Additionally, ASC contains all other classes defined here, except for $\Delta$-substitutes where $\Delta > 1$. For some of these classes, their containment in ASC is non-trivial, and only evident through our demand types analysis in \Cref{sec:demand_types}.

\paragraph{Primitives for accessing $f$.} Since $f: 2^{[n]} \rightarrow [0,1]$ 
is a set function, for our computational results we consider two types of primitives for accessing $f$:

\begin{itemize}
    \item A \emph{value oracle}: Given a set $S\subseteq [n]$, a value query returns the value $f(S)$.
    \item A \emph{demand oracle}: Given a vector of price $p\in \reals^n$, a demand query returns a set $S \in \demandset{f}{p}$.
\end{itemize}

Value and demand oracles are standard both in literature on combinatorial auctions and the literature on combinatorial contracts (see, e.g., \citep{DuettingEFK21,DEFK21-journal,VuongDPP24,multi,multimulti}). It is well known that demand oracle access is a stronger assumption than value oracle access, since a value oracle can be simulated in poly-time with value oracle access, but the reverse is not true \citep{nisanalgorithmic}.

\begin{figure}[t]
\centering
    \begin{subfigure}[t]{0.48\textwidth}
        \centering
\begin{tikzpicture}[scale=0.85]
    \begin{axis}[
        axis x line=bottom,
        axis y line=left,
        axis line style = {-stealth},   
        x axis line style={dash pattern=on 5pt off 5pt, thick},
        xlabel style={at={(axis description cs:1,0)}, anchor=west},
        ylabel style={at={(axis description cs:0,1)}, anchor=south, rotate=-90},
        xlabel = {$\alpha^{-1}$},
        ylabel = {$u_B(\alpha^{-1} c)$},
        xtick={1,
            (\figthreedfxy - \figthreedfy)/\figthreedpx, 
            (\figthreedfy - \figthreedfx)/(\figthreedpy - \figthreedpx), 
            \figthreedfx/\figthreedpx
        },
        ytick=\empty,
        ymin=0,
        xmin=1,
        xmax=3.2, 
        ymax=2.5,
        xmajorgrids,
        legend pos=north east,
        axis background/.style={fill=gray!20},
        x tick label style={font=\footnotesize},
    ]

    \addplot [
        domain=1:(\figthreedfxy - \figthreedfy)/\figthreedpx,
        ultra thick,
        samples=100, 
        color=color12,
    ]
    {\figthreedfxy - x*(\figthreedpx + \figthreedpy)};
    \addlegendentry{$u_B(\alpha^{-1} c, \{1,2\})$}

    \addplot [
        domain=(\figthreedfxy - \figthreedfy)/\figthreedpx:(\figthreedfy - \figthreedfx)/(\figthreedpy - \figthreedpx), 
        ultra thick,
        samples=100, 
        color=color2,
    ]
    {\figthreedfy - x*\figthreedpy};
    \addlegendentry{$u_B(\alpha^{-1} c, \{2\})$}

    \addplot [
        domain=(\figthreedfy - \figthreedfx)/(\figthreedpy - \figthreedpx):\figthreedfx/\figthreedpx, 
        ultra thick,
        samples=100, 
        color=color1,
    ]
    {\figthreedfx - x*\figthreedpx};
    \addlegendentry{$u_B(\alpha^{-1} c, \{1\})$}

    \path (1, 0) -- node[midway, above, color12] {$\{1,2\}$} ({(\figthreedfxy - \figthreedfy)/\figthreedpx}, 0);
    
    \path ({(\figthreedfxy - \figthreedfy)/\figthreedpx}, 0) -- node[midway, above, color2] {$\{2\}$} ({(\figthreedfy - \figthreedfx)/(\figthreedpy - \figthreedpx)}, 0);
    
    \path ({(\figthreedfy - \figthreedfx)/(\figthreedpy - \figthreedpx)}, 0) -- node[midway, above,color1] {$\{1\}$} ({\figthreedfx/\figthreedpx}, 0);
    
    \path ({\figthreedfx/\figthreedpx}, 0) -- node[midway, above] {$\emptyset$} (3.2, 0);

    \end{axis}
\end{tikzpicture}
    \end{subfigure}
    \hfill
    \begin{subfigure}[t]{0.48\textwidth}
        \centering \begin{tikzpicture}[scale=0.85]
    \begin{axis}[
        axis x line=bottom,
        axis y line=left,
        axis line style = {-stealth},
        xlabel style={at={(axis description cs:1,0)}, anchor=west},
        ylabel style={at={(axis description cs:0,1)}, anchor=south, rotate=-90},
        xlabel = {$\alpha$},
        ylabel = {$u_A(\alpha, S_\alpha)$},
        axis background/.style={fill=gray!20},
        xtick={
            \figthreedpx/\figthreedfx, 
            (\figthreedpy - \figthreedpx)/(\figthreedfy - \figthreedfx), 
            \figthreedpx/(\figthreedfxy - \figthreedfy)
        },
        ytick=\empty,
        ymin=0,
        ymax=1.6,
        xmin=0.3,
        xmax=0.85,
        xmajorgrids,
        legend pos=north west,
        x tick label style={font=\footnotesize},
    ]

\pgfmathsetmacro{\aone}{\figthreedpx/\figthreedfx}
\pgfmathsetmacro{\atwo}{(\figthreedpy - \figthreedpx)/(\figthreedfy - \figthreedfx)}
\pgfmathsetmacro{\athree}{\figthreedpx/(\figthreedfxy - \figthreedfy)}

\pgfplotsset{
  notbest/.style={densely dashed, opacity=0.8, ultra thick},
  best/.style={ultra thick},
}

\draw[very thick, black] (0.3, 0) -- node[midway, above] {$\emptyset$} (\aone, 0);

\addplot[
  domain=0.3:\aone,
  samples=50,
  color=color1,
  notbest,
  forget plot,
] {x*\figthreedfx - \figthreedpx};

\addplot[
  domain=\aone:\atwo,
  samples=50,
  color=color1,
  best,
] {x*\figthreedfx - \figthreedpx};
\addlegendentry{$u_A(\alpha, \{1\})$}

\addplot[
  domain=\atwo:0.85,
  samples=50,
  color=color1,
  notbest,
  forget plot,
] {x*\figthreedfx - \figthreedpx};

\addplot[
  domain=0.3:\atwo,
  samples=50,
  color=color2,
  notbest,
  forget plot,
] {x*\figthreedfy - \figthreedpy};

\addplot[
  domain=\atwo:\athree,
  samples=50,
  color=color2,
  best,
] {x*\figthreedfy - \figthreedpy};
\addlegendentry{$u_A(\alpha, \{2\})$}

\addplot[
  domain=\athree:0.85,
  samples=50,
  color=color2,
  notbest,
  forget plot,
] {x*\figthreedfy - \figthreedpy};

\addplot[
  domain=0.3:\athree,
  samples=50,
  color=color12,
  notbest,
  forget plot,
] {x*\figthreedfxy - (\figthreedpx + \figthreedpy)};

\addplot[
  domain=\athree:0.85,
  samples=50,
  color=color12,
  best,
] {x*\figthreedfxy - (\figthreedpx + \figthreedpy)};
\addlegendentry{$u_A(\alpha, \{1,2\})$}

\draw[very thick, color1] (\aone, 0) -- node[pos=0.8, above] {$\{1\}$} (\atwo, 0);
\draw[very thick, color2] (\atwo, 0) -- node[midway, above] {$\{2\}$} (\athree, 0);
\draw[very thick, color12] (\athree, 0) -- node[midway, above] {$\{1,2\}$} (0.85, 0);

    \end{axis}
\end{tikzpicture}

    \end{subfigure}
    \caption{The connection between the auction buyer's utility at prices $\alpha^{-1} \cdot c$ (left, same as \Cref{fig:3d_slice}) and the agent's utility at $\alpha$ (right), for $f(\{1\})=2, f(\{2\})=3, f(\{1,2\})=4$ and $c_1=0.7, c_2=7/6$.}\label{fig:agent_upper_envelope_buyer_util}
\end{figure}
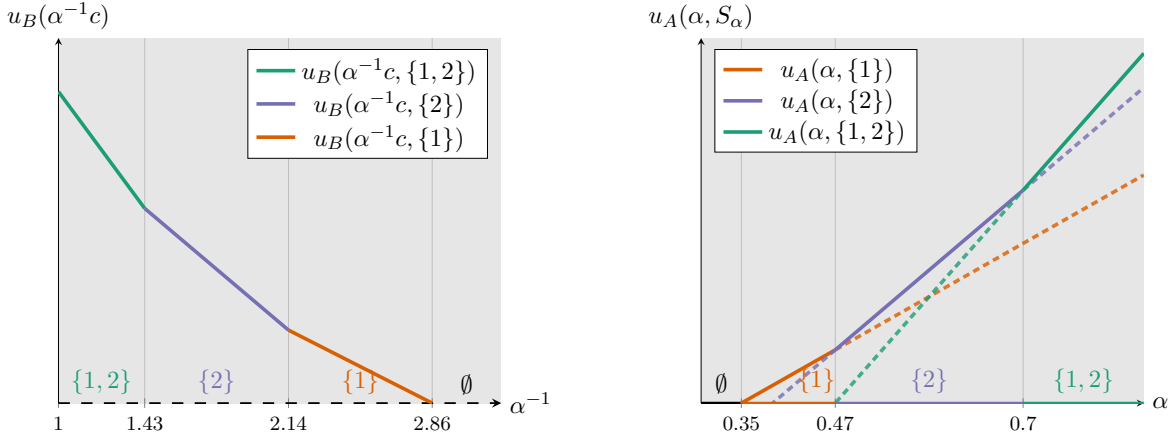

\subsection{Critical Values and Finding an Optimal Contract}\label{sec:crit_values}

We next describe known techniques and algorithmic results for computing an optimal (linear) contract, for different classes of success probability functions.

\paragraph{Upper envelope and critical values.} 
In the analysis of linear contracts, a geometric analysis of the agent's utility, such as that in \citep{simple_vs_optimal}, is useful. Under a subset $S\subseteq [n]$ of actions, the agent's utility as a function of $\alpha$ is $\alpha \cdot f(S) - \sum_{i\in S} c_i$, i.e., it is a linear expression in $\alpha$. When plotting these lines (see Figure~\ref{fig:agent_upper_envelope_buyer_util} (right)), the agent's best response to any linear contract is given by the upper envelope.
We say that a value of $\alpha$ is \emph{critical}, if it's a value of $\alpha$ at which the agent's best response changes. Equivalently, the critical values are the points where two lines on the upper envelope cross.
It is known that any optimal contract is also a critical value.

The geometric view of the optimal contract problem motivates algorithmic approaches that recover an optimal contract by searching over critical points, leveraging the structure they induce. The following proposition identifies sufficient conditions under which this approach can be implemented efficiently.

\begin{proposition}[\citet{DuettingFG24,VuongDPP24}]
\label{prop:eisner-severance}
Given a success probability function $f$ and additive costs $c$, an optimal contract can be computed using polynomially many value queries if:
\begin{enumerate}
\item[(1)] There is a polynomial number of critical values.
\item[(2)] There exists a polynomial-time algorithm with access to value queries that, given $\alpha$, returns the agent's best response (i.e., answers a demand query).
\end{enumerate}
\end{proposition}

Notably, in \citet{DuettingEFK21}, the algorithm for finding optimal contracts required three ingredients, including an explicit procedure for moving between critical points. 
Subsequent work \citet{DuettingFG24,VuongDPP24} showed that this third requirement is redundant by providing a recursive algorithm that enumerates all critical points, using only demand queries, in time polynomial in the number of critical values.
This algorithm is closely related to techniques that appeared earlier in sensitivity analysis of combinatorial optimization problems \citep{Gusfield1980}, where it is known as the Eisner-Severance technique \citep{EisnerSeverance1976}.

\paragraph{Satisfying the ingredients of the critical values approach.} For GS $f$, \citet{DuettingEFK21} show that the number of critical points is bounded by $O(n^2)$. It is known that demand queries for GS can be solved in poly-time using a natural greedy algorithm, thus both ingredients for the critical values approach are satisfied.

A line of work has shown a polynomial bound on the number critical points for supermodular reward functions \citep{DuettingFG24, VuongDPP24}, and ultra reward functions \citep{ultracontracts}. Since these classes admit an efficient demand query algorithm, it follows that both ingredients of the critical values approach are satisfied.

\section{The Demand Type Machinery}\label{sec:demand_types}

In this section, we present the demand type machinery of \citet{demand_types}, and relate it to the classes of set functions defined in \Cref{sec:model}.
Demand types are essentially classes of set functions, defined by sets of vectors, which we term \emph{demand covers}.

In \Cref{sec:demand_types_prelim}, we give the important definitions for understanding demand types, starting with the Locus of Indifference Prices (LIP) --- a key concept in the analysis of \citet{demand_types}. We then present the definitions of demand covers and demand types, and relate them to the geometry of the LIP. 
In \Cref{sec:demand-covers}, we present the known demand covers of the classes gross substitutes, gross complements, and $\Delta$-substitutes, as well as extend this machinery to the new classes of functions considered in this paper.

\subsection{Locus of Indifference Prices and Demand Types}\label{sec:demand_types_prelim}

The high-level idea of the analysis in \citet{demand_types} is to consider the vector space $\reals^n$ as ``price-space'', {and partition this according to the identity of the}
demand set $\demandset{f}{p}$, for a given set function $f: 2^{[n]} \rightarrow \mathbb{R}_{\ge 0}$. This gives rise to a geometric structure, which enables a characterization of possible changes in demand in response to sufficiently small generic price changes. Informally, this characterization is the demand type.

We start by stating the key definitions of \citet{demand_types} for our setting.
{Note that the setting of \citet{demand_types} is strictly more general, in that they} 
allow the domain of the set function $f$ to be any finite subset of $\mathbb{Z}^{[n]}$, while we only consider set functions $f$ with domain $2^{[n]}$. 
We discuss the implications of this divergence 
in Appendix~\ref{app:demand_types}.

\paragraph{Geometry of demand: LIP, UDRs, and their structure.}
Let $f:2^{[n]}\rightarrow \reals_{\ge 0}$ be a set function.
The \emph{Locus of Indifference Prices (LIP)} of $f$ is defined as the price vectors where demand isn't unique, i.e., $\lip_f = \{p\in \reals^n\mid |\demandset{f}{p}| > 1\}$. 
From the continuity of the expression $f(S) - \sum_{i\in S} p_i$ in the vector $p$ it follows
that 
the complement of the LIP consists of $n$-dimensional components in which
demand is unique and constant. These components are called the \emph{Unique Demand Regions (UDRs)}. 
Meanwhile, the LIP itself is the union of $(n-1)$-dimensional linear pieces called \emph{facets}, which separate the UDRs. In the relative interior of a facet (defined as the interior taken in the affine span), demand is indifferent between the two bundles demanded in the UDRs on either side of the facet.

\newcommand*{\TickSize}{2pt}

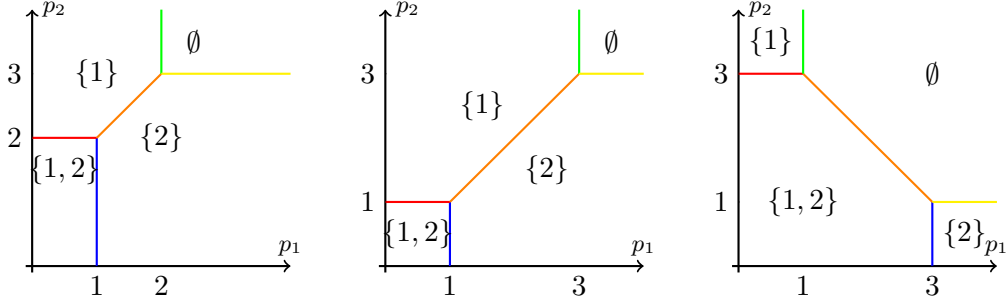
\begin{figure}
\begin{center}
\begin{tikzpicture}[scale=0.85]
    \draw[thick, ->] (-0.1, 0) -- (4, 0) node[above] {\footnotesize $p_1$};
    \draw[thick, ->] (0, -0.1) -- (0, 4) node[right] {\footnotesize $p_2$};
    \draw[blue,thick] (1, 0) -- (1, 2);
    \draw[red,thick] (0, 2) -- (1, 2);
    \draw[orange,thick] (1, 2) -- (2, 3);
    \draw[green,thick] (2, 3) -- (2, 4);
    \draw[yellow,thick] (2, 3) -- (4, 3);
    
    \node at (0.5, 1.5) {$\{1,2\}$};
    \node at (2.5, 3.5) {$\emptyset$};
    \node at (1, 3) {$\{1\}$};
    \node at (2, 2) {$\{2\}$};

    \draw (1,0) -- (1,0) node [below] {$1$};
    \draw (2,0) -- (2,0) node [below] {$2$};
    \draw (0,2) -- (0,2) node [left] {$2$};
    \draw (0,3) -- (0,3) node [left] {$3$};
\end{tikzpicture}
\quad
\begin{tikzpicture}[scale=0.85]
    \draw[thick, ->] (-0.1, 0) -- (4, 0) node[above] {\footnotesize $p_1$};
    \draw[thick, ->] (0, -0.1) -- (0, 4) node[right] {\footnotesize $p_2$};
    \draw[blue,thick] (1, 0) -- (1, 1);
    \draw[red,thick] (0, 1) -- (1, 1);
    \draw[orange,thick] (1, 1) -- (3,3);
    \draw[green,thick] (3, 3) -- (3, 4);
    \draw[yellow,thick] (3, 3) -- (4, 3);

    \node at (0.5, 0.5) {$\{1,2\}$};
    \node at (3.5, 3.5) {$\emptyset$};
    \node at (1.5, 2.5) {$\{1\}$};
    \node at (2.5, 1.5) {$\{2\}$};

    \draw (1,0) -- (1,0) node [below] {$1$};
    \draw (3,0) -- (3,0) node [below] {$3$};
    \draw (0,1) -- (0,1) node [left] {$1$};
    \draw (0,3) -- (0,3) node [left] {$3$};
\end{tikzpicture}
\quad
\begin{tikzpicture}[scale=0.85]
    \draw[thick, ->] (-0.1, 0) -- (4, 0) node[above] {\footnotesize $p_1$};
    \draw[thick, ->] (0, -0.1) -- (0, 4) node[right] {\footnotesize $p_2$};
    \draw[blue,thick] (3, 0) -- (3, 1);
    \draw[red,thick] (0, 3) -- (1, 3);
    \draw[orange,thick] (1, 3) -- (3,1);
    \draw[green,thick] (1, 3) -- (1, 4);
    \draw[yellow,thick] (3, 1) -- (4, 1);

    \node at (1, 1) {$\{1,2\}$};
    \node at (3, 3) {$\emptyset$};
    \node at (0.5, 3.5) {$\{1\}$};
    \node at (3.5, 0.5) {$\{2\}$};

    \draw (1,0) -- (1,0) node [below] {$1$};
    \draw (3,0) -- (3,0) node [below] {$3$};
    \draw (0,1) -- (0,1) node [left] {$1$};
    \draw (0,3) -- (0,3) node [left] {$3$};
\end{tikzpicture}
\end{center}
\caption{Examples of the LIP, when $n=2$. 
{\bf Left ($f_1$)}: $f_1(\emptyset) = 0, f_1(\{1\}) = 2, f_1(\{2\})=3, f_1(\{1,2\})=4$. {\bf Middle ($f_2$)}: $f_2(\emptyset) = 0, f_2(\{1\}) = 3, f_2(\{2\})=3, f_2(\{1,2\})=4$. {\bf Right ($f_3$)}: $f_3(\emptyset)=0, f_3(\{1\})=1, f_3(\{2\})=1, f_3(\{1,2\})=4$} \label{fig:lip_examples}
\end{figure}
{Figure~\ref{fig:lip_examples} shows examples of the LIP for two-item settings, for three different set functions. The colored lines are the facets, and the LIP is their union. The UDRs are labeled with the unique demanded bundle within them.}

Central to the demand types analysis is the following observation by \citet{demand_types}, which says that the change in demand between prices that cross one facet must be given by the vector which is normal to this facet.\footnote{See \citep{demand_types, baldwinworkingpaper}
for more discussion of the properties and economic significance of these vectors.} Denote the characteristic vectors of $S$ and $T$ by $\chi_S, \chi_{T}$, i.e., $(\chi_{S})_i = \indicator[i\in S]$.
\begin{observation}[\citet{demand_types}]\label{obs:change_in_demand_normal}
    Let $F$ be a facet of the LIP, and let $S$ and $T$ be the unique demanded bundles in the UDRs on either side of $F$. 
    The vector $\chi_S-\chi_{T} \in \{-1, 0, 1\}^n$, which represents the change in demand, is normal to $F$. 
\end{observation}

This observation motivates examining the vectors which are normal to facets, therefore, for any set function $f:2^{[n]}\rightarrow \reals_{\ge 0}$, we denote by $V(f)$ the set of all vectors which are normal to a facet of $\lip_f$, which also represents the ways in which demand can change between two adjacent UDRs.

In the leftmost example of Figure~\ref{fig:lip_examples}, the LIP consists of five facets: the red and yellow facets have normal vectors $\pm(0,1)$; the blue and green facets have normal vectors $\pm(1,0)$; and the orange facet has normal vectors $\pm(1,-1)$, thus $V(f_1)=\pm \{(0,1), (1,0), (1,-1)\}$, where $\pm$ signifies taking either the plus or minus of each element in the set.
The middle example has {the same structure, and the same normal vectors, although the precise values that $f_1$ and $f_2$ attach to 
bundles differ.}
The rightmost example has a different structure: the orange facet now has normal vectors $\pm(1,1)$, while the other facet normals remain unchanged, thus $V(f_3)=\pm \{(0,1), (1,0), (1,1)\}$.
As noted above, when crossing a facet, demand changes in the direction of the facet's normal vector. For example, crossing the red facet upwards changes demand from $\{1,2\}$ to $\{1\}$, and the normal vector is indeed $(0,1) = \chi_{\{1,2\}} - \chi_{\{1\}}$. Similarly, moving {left to right}
across the orange facet in the left and middle examples changes demand from $\{1\}$ to $\{2\}$, with normal vector $(1,-1) = \chi_{\{1\}} - \chi_{\{2\}}$.  {Note that this swap between goods may intuitively be understood as a substitution, and corresponds to an instance of our \Cref{def:common-set-functions}(\ref{def:gs})}. {On the other hand, moving left to right across the orange facet in the right-hand example changes demand from $\{1,2\}$ to $\emptyset$, with normal vector $(1,1)=\chi_{1,2}-\chi_\emptyset$.  Demand for both goods reduces in response to a price increase for just one of them because, in this case, the goods are complements (\Cref{def:common-set-functions}(\ref{def:complements})).} 

\paragraph{Demand covers and demand types.}
We now turn to the definitions of demand covers and types:

\begin{definition} [Demand cover and demand-covered functions] \label{def:demand_cover}
    We call a set $\demandcover\subseteq \{-1, 0,1\}^n$ a \emph{demand cover} if for any $v\in \demandcover$ it holds that $-v\in \demandcover$ as well. 
    Let $\demandcover$ be a demand cover and $f:2^{[n]}\rightarrow \reals_{\ge 0}$ be a set function, if $V(f)\subseteq \demandcover$, we say that $\demandcover$ \emph{demand-covers} $f$, and that $f$ is demand-covered by $\demandcover$. 
\end{definition}

\begin{definition} [Demand type] \label{def:demand_type} 
    Let $\demandcover\subseteq \{-1, 0,1\}^n$ be a demand cover. The \emph{demand type} defined by $\demandcover$ is the collection of all 
    set 
    functions $f:2^{[n]}\rightarrow \reals_{\ge 0}$ which are demand-covered by $\demandcover$, and is denoted by $\demandtype(\demandcover)$.
\end{definition}

To illustrate \Cref{def:demand_cover} and \Cref{def:demand_type}, we again refer to \Cref{fig:lip_examples}. 
Recall that for functions $f_1$ and $f_2$ we had $V(f_1)=V(f_2)=\pm \{(1,0), (0,1),(1,-1)\}$, while for function $f_3$ we had 
$V(f_3) = \pm \{(1,0), (0,1), (1,1)\}$. Thus $\pm \{(1,0), (0,1),(1,-1)\}$ demand-covers $f_1$ and $f_2$, and $\pm \{(1,0), (0,1),(1,1)\}$ demand-covers $f_3$, but the only demand cover which covers all $3$ functions is $\{-1,0,1\}^2$, which defines the demand type of all functions $f:2^{[2]}\rightarrow \reals_{\ge 0}$.

An important observation is that not every collection of set functions $f:2^{[n]}\rightarrow \reals_{\ge 0}$ corresponds exactly to some demand type. For example, the collection $\{f_1\}$, \emph{is not} a demand type, because $V(f_1)=V(f_2)$, implying they either both belong to a demand type, or neither does.  

On the other hand, every collection of set functions $f:2^{[n]}\rightarrow \reals_{\ge 0}$ is contained in some demand type, because the demand type defined by the demand cover $\{-1, 0,1\}^n$ is the collection of all set functions $f:2^{[n]}\rightarrow \reals_{\ge 0}$. When examining a collection through the demand type machinery, we are interested in preserving as many of its structural properties as possible, motivating the following definitions.

\begin{definition} [Minimal demand cover / minimal containing demand type] \label{def:min_cover}
    Let $C$ be a collection of set functions $f:2^{[n]}\rightarrow \reals_{\ge 0}$. The \emph{minimal demand cover} of $C$ is denoted by $\mincover(C)$ and is a demand cover $\demandcover$ such that:
    \begin{enumerate}[itemsep=0pt, parsep=0pt, topsep=0pt]
        \item $C\subseteq \demandtype(\demandcover)$.
        \item It holds that $C\not \subseteq \demandtype(\demandcover')$ for any demand cover $\demandcover'\subsetneq \demandcover$. 
    \end{enumerate}
    We call the demand type defined by the the minimal demand cover (i.e., $\demandtype (\mincover(C))$)  the \emph{minimal containing demand type} of $C$.
\end{definition}

Note that an equivalent definition of a minimal demand cover is $\mincover(C) = \bigcup_{f\in C} V(f)$. Thus, $\mincover(C)$ is well defined.  

{Although the minimal demand cover always exists, there is not always a demand cover such that the first condition of \Cref{def:min_cover} is an equality.  This is because, as noted above, not every collection of set functions is a demand type. Nevertheless, it is useful to examine the corresponding objects when they do exist:}
\begin{definition} [Exact demand cover / type] \label{def:exact_cover}
    Let $\demandcover = \mincover(C)$ be the minimal demand cover of a collection $C$ of set functions $f:2^{[n]}\rightarrow \reals_{\ge 0}$. We say that $\demandcover$ is the \emph{exact demand cover} of $C$ if it satisfies $C=\demandtype(\demandcover)$. 
    In this case, we say that $C$ admits an exact demand cover and we refer to $C$ as an \emph{exact demand type}.  
\end{definition}

To illustrate \Cref{def:min_cover} and \Cref{def:exact_cover}, let's return to our earlier discussion about $f_1$ and $f_2$ from \Cref{fig:lip_examples} using the new terminology. The minimal demand cover of $\{f_1\}$ is $\mincover(\{f_1\})=V(f_1)=\pm\{(0,1), (1,0), (1,-1)\}$.
However, $\{f_1\}$ does not 
admit an \emph{exact} demand cover, since $f_2 \in \demandtype(\mincover(\{f_1\}))$. That is, $\{f_1\}$ is not an exact demand type. 
As was observed in \citet{demand_types} (See \Cref{prop:gross_covers}(\ref{item:gc_cover})), the class of gross-substitutes set functions does have an exact demand cover, and in the two dimensional case, it is $\pm\{(0,1), (1,0), (1,-1)\}$. This also implies that gross-substitutes is the minimal containing demand type of $\{f_1\}$. This makes sense, because $f_1$ is a GS set function.  Similarly, $f_3$ is a gross complements set function, and its minimal containing demand type is the set of all gross complements set functions over two items.

\begin{remark} \label{remark:redudant_vectors}
    A collection of set functions may be an exact demand type, i.e., equal to a demand type defined by some demand cover $\demandcover$, without $\demandcover$ being an exact demand cover of it.
    This can happen only in cases where $\demandcover$ is not \emph{minimal}, i.e., some vectors in $\demandcover$ are redundant. 
    Such an example can be found in \Cref{app:demand_types}.  This situation does not arise in the more general setting of \citet{demand_types} when the domain of valuation functions is not restricted to $2^{[n]}$, and so we will need to check minimality even though they do not. 
\end{remark}

\paragraph{Extension to families.} 
Since we study asymptotics, we are not only interested in the demand type or demand cover for any specific $n$, but ``for all $n$s'' at once. In what follows, we therefore appropriately extend the definitions of demand cover and demand type to what we term demand-cover families and demand-type families. 

\begin{definition} [Demand-cover family] \label{def:demand_type_fam}
    Let $\demandcoverfam = \{\demandcover_n\}_{n\in \positivenats}$ be such that 
    for any $n\in \positivenats$, $\demandcover_n \subseteq \{-1, 0,1\}^n$ is a demand cover. We call such  $\demandcoverfam$ a \emph{demand-cover family}.
\end{definition}

\begin{definition} [Demand-type family]
    Let $\demandcoverfam=\{\demandcover_n\}_{n\in \positivenats}$ be a demand-cover family. The demand-type family defined by $\demandcoverfam$ is the class of set functions which are demand-covered by some $\demandcover_n$, and is denoted by $\demandtypefam(\demandcoverfam)$. I.e., $\demandtypefam(\demandcoverfam) = \bigcup_{n\in \positivenats} \demandtype(\demandcover_n)$.  
\end{definition}

Let $\mathcal{C}$ be a class of set functions (e.g., GS, ultra, etc.). 
The definitions of minimal demand cover and exact demand cover are extended to demand-cover families in the natural way, i.e., we say that $\demandcoverfam=\{\demandcover_n\}_{n\in \positivenats}$ is the exact (resp. minimal) demand-cover family of $\mathcal{C}$ if $\demandcover_n$ is the exact (resp. minimal) demand cover of $\{f:2^{[n]}\rightarrow \reals_{\ge 0}\mid f\in \mathcal{C}\}$, for all $n \in \positivenats$.
We denote the minimal demand-cover family of $\mathcal{C}$ by $\mincoverfam(\mathcal{C})$. The definitions of minimal containing demand-type families and exact demand-type families follow.

\subsection{Classes of Set Functions via Demand Types}\label{sec:demand-covers}
In the following, we present demand-cover families of the classes of set functions discussed in this paper. It is this perspective that will drive our results in the following sections. 
Some of these covers were presented in~\citet{demand_types} and~\citet{near_substitutes}.
A subtle point here, is that when taking their results, in order to show exactness we must also show minimality in our setting (see \Cref{remark:redudant_vectors}), this, as well as the proofs for our new demand covers, is done in \Cref{app:demand_covers_minimality}. 

We start with gross substitutes and gross complements. The intuition underlying the following two results is that if a vector in $V(f)$ contains two entries of the same sign, then with any price change crossing the associated facet, demand for the corresponding goods changes in the same way.  This demonstrates, locally, a \emph{complementarity} between these goods.   Similarly, if a vector in $V(f)$ contains two entries of opposite signs, then one of the corresponding goods is \emph{substituted} for the other under any price change crossing the associated facet.

\begin{proposition}[Gross substitutes and gross complements are exact demand-type families \citet{demand_types} together with \Cref{lem:gs_cover_minimality} and \Cref{lem:supermodular_cover_minimality}] \label{prop:gross_covers}
    Gross substitutes and gross complements are exact demand-type families. Furthermore:
    \begin{enumerate}
        \item The exact demand-cover family of gross substitutes is the set of vectors with entries in $\{-1,0,1\}$ with at most one positive coordinate and at most one negative coordinate.\label{item:gs_cover}
        \item The exact demand-cover family of gross complements is the set of vectors with entries in $\{-1,0,1\}$ whose non-zero coordinate entries are all of the same sign.\label{item:gc_cover}
    \end{enumerate}
\end{proposition}

Referring back to \Cref{fig:lip_examples}, we can now clearly see that the two examples on the left are GS, as all the horizontal and vertical facets have normals in the standard basis, and the orange facet has normals $\pm(1,-1)$, therefore they are GS by \Cref{prop:gross_covers}(\ref{item:gs_cover}). The example on the right is gross complements; 
again, all horizontal and vertical facets have normals in the standard basis, and the orange facet has normals $\pm(1,1)$, therefore it is gross complements 
by \Cref{prop:gross_covers}(\ref{item:gc_cover}).

\begin{proposition} [$\Delta$-substitutes is an exact demand-type family \citet{near_substitutes} and  \Cref{lem:delta_sub_cover_min}]\label{prop:near-substitutes_demand_type}
    For any $\Delta\in \mathbb{N}^+$, $\Delta$-substitutes is an exact demand-type family, whose exact demand-cover family is the set of vectors with entries in $\{-1, 0, 1\}$ with at most $\Delta$ positive coordinates, and at most $\Delta$ negative coordinates.
\end{proposition}

We remark that, as expected, the exact demand-cover family of $1$-substitutes is the same as the exact demand-cover family of GS, and thus the two classes are equal.
Additionally, this demand cover fits the way we defined $\Delta$-substitutes(\Cref{def:uncommon-set-functions}(\ref{def:delta_sub})) --- the vectors in the cover are exactly the ways in which demand changes in response to small changes in a single item's price per the definition of $\Delta$-substitutes. This is equivalent for demand types of set functions (see \Cref{lem:axis_parallel_equivalent}), and the reason why the definition we provide is equivalent to that given by \citet{near_substitutes}.

We next turn to the classes of gross substitutes and complements (GSC) \citep{sun_and_yang} and GSC+. And show that GSC+ is the minimal containing demand type of GSC. \Cref{ex:gsc+_not_gsc} shows this containment is strict.

\begin{restatable}{proposition}{gscplus}[GSC+ is the minimal containing demand-type family of GSC \citep{demand_types} and  \Cref{app:gsc+}] 
\label{prop:gsc+}
    GSC+ is an exact demand-type family, with the exact demand-cover family given by the set of vectors with entries in $\{-1,0,1\}$ with at most $2$ non-zero coordinates. 
    Furthermore, GSC+ is the minimal containing demand-type family of GSC.
\end{restatable}

\paragraph{A broader demand type: ASC}
We now state the exact demand-cover family for \gooddemand, which is a key part of our analysis. 

\begin{proposition}[\gooddemand is an exact demand-type family, \Cref{app:asc_exact}] \label{prop:gooddemand_type}
    \gooddemand is an exact demand-type family with the exact demand-cover family given by vectors $v$ with entries in $\{-1, 0,1\}$ satisfying one of the following:
    \begin{enumerate}
        \item all non-zero coordinates of $v$ have the same sign. 
        \item $v$ has at exactly one positive entry and exactly one negative entry.
    \end{enumerate}
\end{proposition}
Observe that, as can be seen in \Cref{prop:gooddemand_type}, the exact demand-cover family of ASC contains that of gross substitutes, gross complements and GSC+, and thus ASC contains all of them.

We now connect ASC to ultra (\Cref{def:our-set-functions}(\ref{def:ultra})), the proof is deferred to \Cref{sec:asc_min_ultra}.

\begin{restatable}{proposition}{ultrademandcover}\label{prop:ultra_demand}
    \gooddemand is the minimal containing demand-type family of ultra.
\end{restatable}

We also show in \Cref{sec:venn_diagram_examples} that \gooddemand is not only a superclass of the other classes we discuss, it is in fact, greater than their union.
\begin{lemma}
    ASC is a strict superclass of the union ultra $\cup$ gross complements $\cup$ GSC+.
\end{lemma}

\section{Poly-Many Critical Points for \gooddemand\ Demand Type}\label{sec:poly_many_result}
In this section, we present our main result, stated in the following theorem.

\begin{theorem} \label{thm:gooddemand_poly_crit}
When $f$ is \gooddemand, the number of critical points is $O(n^2)$.
\end{theorem}

This theorem generalizes all previously known results guaranteeing at most polynomially many critical points --- specifically for the classes of gross substitutes~\citep{DuettingEFK21}, ultra~\citep{ultracontracts}, and supermodular functions~\citep{VuongDPP24, DuettingFG24}. Since each of these classes admits a polynomial-time implementation of demand queries via value queries, this implies (by \Cref{prop:eisner-severance}) a polynomial-time algorithm for computing the optimal contract in all these cases.
However, \Cref{thm:gooddemand_poly_crit} is strictly more general than these cases (recall \Cref{fig:asc_classes_venn}).

We start, in \Cref{sec:lin_cont_line} by interpreting critical points through the demand-type machinery of \citet{demand_types}, using the terminology set up in \Cref{sec:demand_types}. This interpretation sets up the framework for the proof of \Cref{thm:gooddemand_poly_crit}, which is presented in \Cref{sec:poly_crit_proof}

\subsection{Analyzing Critical Points Through Demand Types}\label{sec:lin_cont_line}

Throughout this section we assume that $f$ and $c$ are such that the agent's best response is non-constant.
Recall that under a linear contract $\alpha\in(0,1]$, the agent's best response problem is to find $S$ such that $\alpha f(S) - c(S) = \alpha (f(S) - \sum_{j \in S} c_j/\alpha)$ is maximized. Equivalently, we seek to solve a demand query with the set function $f$ at prices $\left(\frac{c_1}{\alpha}, \dots, \frac{c_n}{\alpha}\right)=\alpha^{-1} c$.

This leads us to define the \emph{linear contract ray}
\[
\ray{c} = \left\{\alpha^{-1} c \mid \alpha \in (0,1) \right\},
\]
which is exactly the ray given by $\{(c_1\cdot x, \dots, c_n\cdot x)\mid x\in (1, \infty)\}$. 
This is the set of prices we may feed into a demand query to find out the agent's best response to a linear contract.

To avoid having to deal with ties, 
we  consider ``facet-piercing'' cost vectors, defined as follows.

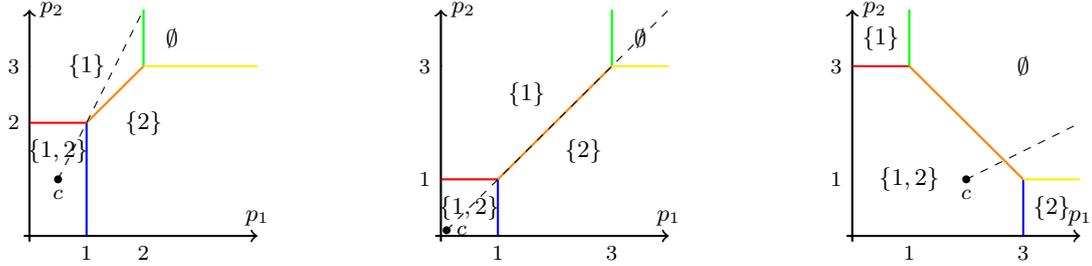
\begin{figure}
    \centering
    
    \begin{subfigure}[t]{0.32\textwidth}
        \centering
        \begin{tikzpicture}[scale=0.75]
            \draw[thick, ->] (-0.1, 0) -- (4, 0) node[above] {\footnotesize $p_1$};
            \draw[thick, ->] (0, -0.1) -- (0, 4) node[right] {\footnotesize $p_2$};
            
            \draw[blue,thick] (1, 0) -- (1, 2);
            \draw[red,thick] (0, 2) -- (1, 2);
            \draw[orange,thick] (1, 2) -- (2, 3);
            \draw[green,thick] (2, 3) -- (2, 4);
            \draw[yellow,thick] (2, 3) -- (4, 3);
            
            \node at (0.5, 1.5) {\footnotesize $\{1,2\}$};
            \node at (2.5, 3.5) {\footnotesize $\emptyset$};
            \node at (1, 3) {\footnotesize $\{1\}$};
            \node at (2, 2) {\footnotesize $\{2\}$};

            \draw (1,0) -- (1,0) node [below] {\scriptsize $1$};
            \draw (2,0) -- (2,0) node [below] {\scriptsize $2$};
            \draw (0,2) -- (0,2) node [left] {\scriptsize $2$};
            \draw (0,3) -- (0,3) node [left] {\scriptsize $3$};

            \draw[dashed] (0.5, 1) -- (2,4);
            \fill (0.5, 1) circle (2pt) node[below] {\scriptsize $c$};
        \end{tikzpicture}
        \caption{
        $f(\{1\})=2, f(\{2\})=3$, \\
        $f(\{1,2\})=4, c=(0.5, 1)$. \\
        The facet-piercing condition is violated at $(1,2)$, where $\ray{c}$ hits the red, orange, and blue facets at once.}
    \end{subfigure}
    \hfill
    \begin{subfigure}[t]{0.32\textwidth}
        \centering
        \begin{tikzpicture}[scale=0.75]
            \draw[thick, ->] (-0.1, 0) -- (4, 0) node[above] {\footnotesize $p_1$};
            \draw[thick, ->] (0, -0.1) -- (0, 4) node[right] {\footnotesize $p_2$};
            
            \draw[blue,thick] (1, 0) -- (1, 1);
            \draw[red,thick] (0, 1) -- (1, 1);
            \draw[orange,thick] (1, 1) -- (3,3);
            \draw[green,thick] (3, 3) -- (3, 4);
            \draw[yellow,thick] (3, 3) -- (4, 3);

            \node at (0.5, 0.5) {\footnotesize $\{1,2\}$};
            \node at (3.5, 3.5) {\footnotesize $\emptyset$};
            \node at (1.5, 2.5) {\footnotesize $\{1\}$};
            \node at (2.5, 1.5) {\footnotesize $\{2\}$};

            \draw (1,0) -- (1,0) node [below] {\scriptsize $1$};
            \draw (3,0) -- (3,0) node [below] {\scriptsize $3$};
            \draw (0,1) -- (0,1) node [left] {\scriptsize $1$};
            \draw (0,3) -- (0,3) node [left] {\scriptsize $3$};

            \draw[dashed] (0.1,0.1) -- (4,4);
            \fill (0.1, 0.1) circle (2pt) node[right] {\scriptsize $c$};
        \end{tikzpicture}
        \caption{
        $f(\{1\})=3, f(\{2\})=3$, \\ 
        $f(\{1,2\})=4,c=(0.1, 0.1)$. \\
        The facet-piercing condition is violated at the endpoints of the orange facet: $\ray{c}$ hits multiple facets.}
    \end{subfigure}
    \hfill
    \begin{subfigure}[t]{0.32\textwidth}
        \centering
        \begin{tikzpicture}[scale=0.75]
            \draw[thick, ->] (-0.1, 0) -- (4, 0) node[above] {\footnotesize $p_1$};
            \draw[thick, ->] (0, -0.1) -- (0, 4) node[right] {\footnotesize $p_2$};
            \draw[blue,thick] (3, 0) -- (3, 1);
            \draw[red,thick] (0, 3) -- (1, 3);
            \draw[orange,thick] (1, 3) -- (3,1);
            \draw[green,thick] (1, 3) -- (1, 4);
            \draw[yellow,thick] (3, 1) -- (4, 1);

            \node at (1, 1) {\footnotesize $\{1,2\}$};
            \node at (3, 3) {\footnotesize $\emptyset$};
            \node at (0.5, 3.5) {\footnotesize $\{1\}$};
            \node at (3.5, 0.5) {\footnotesize $\{2\}$};

            \draw (1,0) -- (1,0) node [below] {\scriptsize $1$};
            \draw (3,0) -- (3,0) node [below] {\scriptsize $3$};
            \draw (0,1) -- (0,1) node [left] {\scriptsize $1$};
            \draw (0,3) -- (0,3) node [left] {\scriptsize $3$};
            \draw[dashed] (2, 1) -- (4,2);
            \fill (2, 1) circle (2pt) node[below] {\scriptsize $c$};
        \end{tikzpicture}
        \caption{
        $f(\{1\})=1, f(\{2\})=1$,\\ 
        $f(\{1,2\})=4, c=(2,1)$.\\
        The facet-piercing condition holds, as there are no multi-facet intersections between $\ray{c}$ and the LIP.}
    \end{subfigure}

    \caption{Examples of the linear contract ray relative to the facet-piercing property. The colored lines represent the LIP facets, the black dot represents $c$, and the dashed line represents $\ray{c}$. (a) Violation due to intersection with red, orange, and blue facets. (b) Violation at the endpoints of the orange facet. (c) A case where the condition holds.}
    \label{fig:facet_piercing}
\end{figure}
\begin{definition} [Facet-piercing vector] 
\label{def:facet-piercing} A vector $c\in \reals^n$ 
is said to be \emph{facet-piercing} with respect to a success probability function $f$ if $\ray{c}$ intersects with the LIP one facet at a time.
\end{definition}

Figure~\ref{fig:facet_piercing} shows two examples of the facet-piercing condition being violated, and one where it is satisfied. The colored lines are the LIP facets, the black dot is $c$, and the dashed line represents the linear contract ray in both instances.
Note that 
a cost vector being facet-piercing implies that the intersection between $\ray{c}$ and the LIP is $0$-dimensional.  

When $c$ is facet-piercing, the agent's best response and critical points are defined by the linear contract line, as follows. For values of $\alpha$ where the linear contract line does not intersect with the LIP, the best response corresponds to the UDR in which the line lies. Critical points occur precisely when the linear contract line intersects a facet of the LIP, and at such points, the agent's best response changes according to the vector normal to the facet being crossed (by \Cref{obs:change_in_demand_normal}).  
To illustrate, consider the rightmost example in \Cref{fig:facet_piercing}. When $\alpha$ is small, the linear contract line lies within the UDR $\emptyset$, which is therefore the agent's best response. There is one critical point --- where the linear contract line crosses the orange facet --- beyond which the best response becomes $\{1,2\}$. The vector normal to the orange facet is $(1,1)$, which indeed equals $\chi_{\{1,2\}} - \chi_{\emptyset}$.

The following two lemmas show that, for the purpose of upper bounding the number of critical values, we can restrict attention to facet-piercing cost vectors. We defer their proofs to \Cref{app:main_result}.

\begin{restatable}{lemma} {facetpiercingmeasure} \label{lemma:facet_piercing_measure_0}
        For any success probability function $f$, the set of cost vectors $c$ that are not facet-piercing is of measure $0$.
\end{restatable}

\begin{restatable}{lemma} {almostallcostswlg} \label{lemma:almost_all_costs_wlg}
     Let $f:2^{[n]}\rightarrow \reals_{\ge 0}$ be a success probability function and $U\subseteq \reals_{\ge 0}^n$ be a set of cost vectors which captures all but a measure $0$ set of cost vectors. For any cost vector $c\in \reals_{\ge 0}^n$, there exists a cost vector $\hat{c} \in U$ such that there are at least as many critical points under $f$ and $\hat{c}$ as there are under $f$ and $c$.
\end{restatable}

\subsection{Applying the Demand Type Framework for \gooddemand}\label{sec:poly_crit_proof}
The proof of \Cref{thm:gooddemand_poly_crit} also relies on the well-known property of increasing best-response costs. For completeness, the proof can be found in \Cref{app:cost_increasing}.
\begin{restatable} {lemma} {brcostmonotone} \label{lem:cost_incr_with_alpha}
For $\alpha_2>\alpha_1$ it holds that $c(S_{\alpha_2})\ge c(S_{\alpha_1})$.
\end{restatable}

We are now ready to prove \Cref{thm:gooddemand_poly_crit}, applying techniques from \citet{DuettingEFK21} alongside the demand type framework.
\begin{proof} [Proof of \Cref{thm:gooddemand_poly_crit}]
    Let $f:2^{[n]}\rightarrow \reals_{\ge 0}$ be an \gooddemand success probability function, and let $c\in \reals_{\ge 0}^n$ be some cost vector. 
    Note that the set of cost vectors with non-unique costs (i.e., there exist $i,j$ such that $c_i=c_j$) is a measure $0$ set. Thus,
    By \Cref{lemma:almost_all_costs_wlg} and \Cref{lemma:facet_piercing_measure_0}, it is without loss of generality to assume that the $c$ is facet-piercing and costs are unique (the union of two measure $0$ sets is also measure $0$). 
    We reorder the actions such that $c_1 < c_2 < \dots < c_n$.  Observe that, after reordering, $f$ is still \gooddemand and $c$ is still facet-piercing.

    The proof will now follow a potential argument. Consider the potential function  $\Phi:2^{[n]}\rightarrow \mathbb{N}$ from \citet{DuettingEFK21}, defined as 
    \[
    \Phi(S) = \sum_{i\in S} i.
    \]
    $\Phi$ is upper bounded by $\Phi([n])=\sum_{i=1}^n i = \frac{n(n+1)}{2}$. We show that it strictly increases at each critical point, concluding the proof.
    Let $\alpha\in(0,1)$ be some critical point and let $S_1$ be the agent's best response directly below $\alpha$, and $S_2$ be the agent's best response directly above it.
    Because $c$ is facet-piercing, $v=\chi_{S_2} - \chi_{S_1}\in V(f)$, and therefore satisfies one of the conditions in \Cref{prop:gooddemand_type}.
    
    If $v$ satisfies condition $(1)$ in \Cref{prop:gooddemand_type}, i.e., all non-zero coordinates of $v$ have the same sign, then either $S_1 \subsetneq S_2$ or $S_2\subsetneq S_1$. Since the cost of the best response is increasing with $\alpha$
    (\Cref{lem:cost_incr_with_alpha})
    , we have $S_1\subsetneq S_2$, and thus $\Phi(S_2) > \Phi(S_1)$, as needed.

    If $v$ satisfies condition $(2)$ in \Cref{prop:gooddemand_type}, i.e., $v$ has exactly one positive entry and exactly one negative entry, then there exist $a\in S_1\setminus S_2$ and $b\in S_2\setminus S_1$ such that $S_2 = (S_1 \setminus \{a\})\cup \{b\}$. Once again applying the monotonicity of best response cost we get $c_b> c_a$ and therefore $b>a$. Thus, $\Phi(S_2)-\Phi(S_1) = b-a > 0$, as needed.
\end{proof}
\section{Computing Demand Queries}
One way to efficiently compute a demand query with value queries is by a specialized greedy algorithm, as done for GS \citep{gross_substitutability} and ultra \citep{ultra_valuations}. In this section, we offer an alternative approach, by utilizing the demand type perspective. We thus obtain an efficient polynomial-time algorithm for computing a demand query (with value queries only) for $\Delta$-substitutes where $\Delta$ is constant, GSC+ and GSC.
Our approach works, beyond these, for any \emph{succinct} demand-type family.
\begin{definition}[Succinct demand-type family]\label{def:succinct}
    Let $\mathcal{C}$ be demand-type family with the exact demand cover family $\demandcoverfam = \{\demandcover_n\}_{n\in \positivenats}$. $\mathcal{C}$ is called \emph{succinct} if there is a polynomial-time algorithm that receives $n$ and outputs $\demandcover_n$.
\end{definition}

\begin{theorem}\label{thm:demand_query}
    Let $\mathcal{C}$ be a succinct demand type family. For any $f\in \mathcal{C}$ and price vector $p\in \reals^n$, there exist a poly-time algorithm that given $p$ and value oracle access to $f$, outputs a set $S\in \demandset{f}{p}$. 
\end{theorem}

Note that $\Delta$-substitutes is a succinct demand-type family, since one can easily compute a poly-size demand cover $\demandcover_n$ for any $n$, as given by \Cref{prop:near-substitutes_demand_type}. We thus get the following corollary, which may be of independent interest. 

\begin{corollary}\label{cor:gsc+_demand_query}
    For any constant $\Delta>0$, a demand query $\Delta$-substitutes (and thus, GSC+ and GSC, which are subclasses of $2$-substitutes) can be computed in polynomial time with value oracle access.
\end{corollary}

\Cref{thm:demand_query} and \Cref{thm:gooddemand_poly_crit} imply the following new result for computing an optimal contract.
\begin{corollary} \label{cor:opt_computation}
    Let $\mathcal{C}$ be a 
    succinct demand-type family contained in ASC.
    For any $f\in \mathcal{C}$, the optimal contract can be computed in polynomial time with value oracle access to $f$.
\end{corollary}
In particular, \Cref{cor:opt_computation} implies that we can efficiently compute the optimal contract for reward functions in GSC+, and more generally, the intersection of $\Delta$-substitutes and ASC.

\paragraph{Algorithm overview.}
The idea of our algorithm (\Cref{alg:demand_query}) is to trace a path in price-space, starting from a point where $\emptyset$ is demanded, and ending at $p$, updating the ``current'' demand bundle as we move along. More concretely, the path we trace is comprised of $n$ steps which are parallel to an axis. In the first step, we change the first coordinate to $p_1$, in the second, we change the second coordinate to $p_2$, and so on, see \Cref{fig:demand_query_path} for an example of such a path. As we show, moving in parallel to an axis ensures that our new demand bundle is an ``adjacent bundle'' of our old demand bundle, in the sense that it belongs to an adjacent UDR. This key property allows us to easily update our demand bundle as we get closer and closer to $p$,  in each step of the path.

\begin{figure}
\begin{center}
\begin{tikzpicture}[scale=0.85]
    \draw[thick, ->] (-0.1, 0) -- (4, 0) node[above] {\footnotesize $p_1$};
    \draw[thick, ->] (0, -0.1) -- (0, 4) node[right] {\footnotesize $p_2$};
    \draw[blue,thick] (2, 0) -- (2, 1);
    \draw[red,thick] (0, 1) -- (2, 1);
    \draw[orange,thick] (2, 1) -- (3, 2);
    \draw[green,thick] (3, 2) -- (3, 4);
    \draw[yellow,thick] (3, 2) -- (4, 2);
    
    \node at (1, 0.5) {$\{1,2\}$};
    \node at (3.5, 2.5) {$\emptyset$};
    \node at (3, 1) {$\{1\}$};
    \node at (2, 2) {$\{2\}$};

    \draw (0,1) -- (0,1) node [left] {$1$};
    \draw (0,2) -- (0,2) node [left] {$2$};
    \draw (2,0) -- (2,0) node [below] {$2$};
    \draw (3,0) -- (3,0) node [below] {$3$};
    \fill (3.5, 3) circle (2pt) node[above] {\scriptsize $q^0$};
    \draw[dashed] (3.5, 3) -- (0.5, 3);
    \draw[dashed] (0.5, 3) -- (0.5, 0.8);
    \fill (0.5, 0.8) circle (2pt) node[left] {\scriptsize $q$};
\end{tikzpicture}
\end{center}
\caption{A two-item example of the price-space path traced by our demand query algorithm, displayed over the LIP. To compute a demand for price vector $q$, our path starts at $q^0$, which lies in the UDR of $\emptyset$, ends at $q$, and is made of $n=2$ segments, each parallel to an axis.
}\label{fig:demand_query_path}
\end{figure}
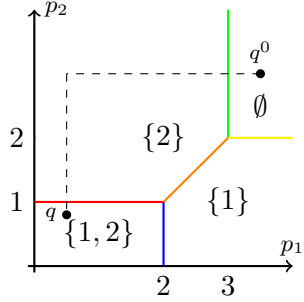

\begin{algorithm}[t]
\caption{Demand query for succinct demand types}\label{alg:demand_query}
\KwIn{A demand cover $\demandcover_n$, value oracle to a function $f\in \demandtype(\demandcover_n)$, and a price vector $p\in \reals^n$.}
\KwOut{A set $S\in \demandset{f}{p}$.}
\tcp{For the definition of $\neighbors{\demandcover_n}{S}$ see \Cref{def:demand_query_notation}}
$M \gets \max_{S\in \neighbors{\demandcover_n}{\emptyset}~\text{s.t.}~ S\ne \emptyset} f(S)$\;
Let $p^0\in \reals^n$ defined by $p^0 \gets (M, M, \dots, M)$\;
Let $S_0 = \emptyset$\;
\For{$i \gets 1$ \KwTo $n$}{
    Let $p^i\in\reals^n$ be defined by $p^i_j \gets \begin{cases}
        p_j & j\le i \\
        M &\text{else}\\
    \end{cases}$\;
    Let $S_i \in \argmax_{T\in \neighbors{\demandcover_n}{S_{i-1}}} f(T) - \sum_{j\in T} p^i_j$\;
}
\Return $S_n$\;
\end{algorithm}

We start with useful notations and key observations. 
\paragraph{Adjacent bundles. } A key component of our algorithm is adjacent bundles; these are the candidates for the updated demanded bundle at each step.
\begin{definition}[Adjacent Bundles]\label{def:demand_query_notation}
The set of adjacent bundles of a bundle $S\subseteq [n]$ with respect to a demand-cover $\demandcover_n$ is defined as 
\[
\neighbors{\demandcover_n}{S}=\{T\subseteq [n]\mid \chi_T-\chi_S \in \demandcover_n \} \cup \{S\}.
\]
This is the set of bundles with UDRs which may be adjacent to $S$ according to $\demandcover_n$, and for simplicity of presentation we include $S$ itself.
\end{definition}
Observe that $\neighbors{\demandcover_n}{S}$ is easily computable given $\demandcover_n$ and $S$. The importance of the adjacent bundles is from the following two properties.
Firstly, as was shown in \citet{baldwinthesis}, a bundle is demanded if and only if it is a local maximum of the buyer's utility, among its adjacent bundles.
\begin{lemma}[\citet{baldwinthesis} Theorem~4.5]\label{lem:local_max_global_demand}
    Let $\demandcover_n$ be a demand cover, $f\in \demandtype(\demandcover_n)$ and $p\in \reals^n$ be a price vector. For any set $S\subseteq [n]$, it holds that $S\in \demandset{f}{p}$ if and only if $u_B(p, S) \ge u_B(p, S')$ for all adjacent bundles $S'\in \neighbors{\demandcover_n}{S}$.
\end{lemma}

Secondly, when we only change one coordinate of a price vector, an adjacent demand bundle to a previously demanded bundle is demanded at updated prices.

\begin{lemma}\label{lem:axis_parallel_demand_sets}
    Let $\demandcover_n$ be a demand cover. A set function $f$ belongs to $\demandtype(\demandcover_n)$ if and only if for any two vectors $p,q\in \reals^n$ which differ by only one coordinate, and for any $S\in \demandset{f}{p}$, there exists a set $T\in \demandset{f}{q}$ such that $T\in \neighbors{\demandcover_n}{S}$.
\end{lemma}
This lemma is a generalization of a result \citet{baldwinworkingpaper} which shows it for the case where both price vectors lie in a UDR.
\begin{theorem}[Theorem~1 of \citep{baldwinworkingpaper}]\label{thm:baldwin_singletons}
    Let $\demandcover_n$ be a demand cover. A set function $f$ belongs to $\demandtype(\demandcover_n)$ if and only if for any two vectors $p,q\in \reals^n$ which differ by only one coordinate and such that $\demandset{f}{p}=\{S\}$ and $\demandset{f}{p}=\{T\}$ it holds that $T\in \neighbors{\demandcover_n}{S}$.
\end{theorem}
We use \Cref{thm:baldwin_singletons} to prove \Cref{lem:axis_parallel_demand_sets}.
\begin{proof} [Proof of \Cref{lem:axis_parallel_demand_sets}]
Let $\demandcover_n$ be a demand cover, and $f:2^{[n]}\rightarrow \reals_{\ge 0}$ be a set function. We will show a logical implication on both directions, thus proving the if and only if. One direction is simple -  if $f$ satisfies the adjacent bundle condition of \Cref{lem:axis_parallel_demand_sets} then it satisfies the weaker condition of \Cref{thm:baldwin_singletons}, and thus $f\in \demandtype(\demandcover_n)$.

For the other direction, assume that $f\in \demandtype(\demandcover_n)$, let $p,q\in \reals^n$ be two price vectors which differ by only one coordinate, and let $S\in \demandset{f}{p}$. By continuity of the buyer's utility, there exists $\varepsilon > 0$ such that any $u\in \reals^n$ which satisfies $\lVert u \rVert \le \varepsilon$ also satisfies $\demandset{f}{q+u} \subseteq \demandset{f}{q}$ and $\demandset{f}{p+u}\subseteq \demandset{f}{p}$. Consider the intersection of the UDR of $S$ and the $\varepsilon$-radius ball around $p$, denote by $\mathcal{R}$. Observe that starting at $p$, decreasing the prices of elements in $S$ and increasing the prices of elements not in $S$ leads to a point in the UDR of $S$, thus the region $\mathcal{R}$ has a positive measure. Moreover, for any vector $u\in \mathcal{R}$ it holds that $\demandset{f}{u}=\{S\}$ and $\demandset{f}{u+(q-p)} \subseteq \demandset{f}{q}$. Since the LIP is of measure $0$ (it is the finite union of $n-1$-dimensional pieces), there exists some $u\in \mathcal{R}$ such that $u+(q-p)\notin \lip_f$, and thus $\demandset{f}{u+(q-p)}=\{T\}$ for some $T\subseteq [n]$. We can now invoke \Cref{thm:baldwin_singletons} on $u$ and $u+(q-p)$, to see that $T\in \neighbors{\demandcover_n}{S}$. Finally, it also holds that $T\in \demandset{f}{u+(q-p)}\subseteq \demandset{f}{q}$, as needed. 
\end{proof}

\paragraph{Putting it all together.} The correctness of \Cref{alg:demand_query} is now evident from \Cref{lem:local_max_global_demand} and \Cref{lem:axis_parallel_demand_sets}.
\begin{theorem}\label{thm:demand_query_alg}
    Let $\demandcover_n$ be a demand cover, $f\in \demandtype(\demandcover_n)$, and $p\in \reals^n$. Given $\demandcover_n$, $p$, and value oracle access to $f$, \Cref{alg:demand_query} runs in polynomial time (in $|\demandcover_n|$ and $n$) and outputs a set $S\in \demandset{f}{p}$. 
\end{theorem}
\begin{proof}
    It is easy to see that \Cref{alg:demand_query} runs in polynomial time in $n$ and $|\demandcover_n|$, since $|\neighbors{\demandcover_n}{S}|\le|\demandcover_n|+1$ for all $S\subseteq [n]$.
    To see that \Cref{alg:demand_query} outputs a set in $\demandset{f}{p}$, we inductively show that for all $i\in \{0,\dots, n\}$, $S_i \in \demandset{f}{p^i}$.
    The induction base is by our choice of $M$, which is such that $u_B(p^0, S) = f(S) - M\cdot |S| \le 0 = u_B(p^0, \emptyset)$ for all $S\in \neighbors{\demandcover_n}{\emptyset}$, and by \Cref{lem:local_max_global_demand}.
    The induction step is by \Cref{lem:axis_parallel_demand_sets}.
\end{proof}

\bibliographystyle{abbrvnat}
\bibliography{bib.bib}

\appendix
\section{Disclosing AI and LLM Usage}

We disclose the use of Artificial Intelligence (AI) and Large Language Models (LLMs) in the preparation of this manuscript as follows:
\begin{itemize}
    \item[(1)] Gemini 3 and ChatGPT 5.2 were used to write an initial draft of the introduction based on bullets that we provided. No BibTeX files or citations were done with generative AI.
    \item[(2)] ChatGPT 5.2 was used to refine details of all figures in the paper.
    \item[(3)] We used Refine.ink and DeepThink to identify some typographical and calculation errors. While these corrections helped refine our calculations and wording of statements and proofs, they did not alter their structure or conclusions.
\end{itemize}

\section{Separating Examples for Classes of Functions}\label{sec:venn_diagram_examples}
In this appendix we present examples that show all regions of the Venn diagram presented in \Cref{fig:asc_classes_venn} are non-empty. 
As for the containment relations in \Cref{fig:asc_classes_venn}, it is known that GS is contained in ultra and GSC \citep{ultra_valuations,sun_and_yang}. All functions are contained in ASC by our demand types analysis in \Cref{sec:demand_types}.

\begin{lemma} [ASC $\setminus ($ Ultra $\cup$ Supermodular $\cup$ GSC+$)$] \label{lem:asc_greater}
    There exists a set function $f:2^{[3]}\rightarrow \reals_{\ge 0}$ which is ASC but is not ultra, supermodular, or GSC+.
\end{lemma}
\begin{proof}
Consider the function $f:2^{[3]} \rightarrow \reals_{\ge0}$ defined by $f(\{1\})=f(\{2\})=f(\{3\})=1$, $f(\{1,2\})=f(\{2,3\})=1.5, f(\{1,3\})=2, f(\{1,2,3\})=10$. By \Cref{lem:normal_vec_price_vec}, we need to show that there exists no vector of prices such that $\demandset{f}{p} = \{\{i\}, [3]\setminus \{i\}\}$ for some $i\in [3]$ and that it's not GSC+. This suffices since clearly $f$ is not supermodular, and it's not ultra since, when taking $S=\{2\}, T=\{1,3\}$ no exchange  yields a non-decreasing sum of valuations.
Assume towards contradiction that there exists a price vector $p$ such that $\demandset{f}{p}=\{\{i\}, [3]\setminus \{i\}\}$, in particular $u_B(p,\{i\})=u_B(p, [3]\setminus \{i\})\ge 0$, but then, since $f([3]) > f(\{i\}) + f([3]\setminus \{i\})$, it holds that $[3]\in \demandset{f}{p}$ as well, contradiction.

Additionally, by \Cref{lem:normal_vec_price_vec}, $p=(10/3, 10/3, 10/3)$ shows that $(1,1,1)$ is a demand type vector, thus $f$ isn't GSC+.
\end{proof}
\begin{lemma}[Ultra $\setminus($ GSC $\cup$ Supermodular$)$] \label{lem:ultra_solitary_region_ex}
    There exists a set function $f:2^{[3]}\rightarrow \reals_{\ge 0}$ which is ultra but is not GSC or supermodular.    
\end{lemma}
\begin{proof}
    Consider the symmetric function $f:2^{[3]}\rightarrow \reals_{\ge 0}$ defined by 
    \[
    f(S) = \begin{cases}
        0 & \text{if }|S|=0 \\
        0 & \text{if }|S|=1 \\ 
        1 & \text{if }|S|=2 \\
        1 & \text{if }|S|=3.
    \end{cases}
    \]
    As was noted by \citet{ultra_valuations}, any symmetric function is ultra, and therefore $f$ is ultra. 
    $f$ is not supermodular, since, for example $f(3\mid \{1,2\}) < f(3\mid \{1\})$.
    It remains to show that $f$ is not GSC.
    To do this we show that any two items display complementarities and so no partition as in~\Cref{def:uncommon-set-functions}(\ref{def:gsc}) exists. Let $i\ne j \in \{1,2\}$. Consider the price vector $p\in \reals^3$ defined by $p_k = \begin{cases}
        2 & \text{if } k\notin \{i,j\} \\
        \frac{1}{2} & \text{else}.
    \end{cases}$
    It holds that $\demandset{f}{p} = \{\emptyset, \{i,j\}\}$, thus proving that $e_i+e_j\in V(f)$ by \cref{lem:normal_vec_price_vec}, as needed.
\end{proof}
\begin{lemma}[GSC $\setminus($ Supermodular $\cup$ Ultra$)$]
    There exists a set function $f:2^{[3]}\rightarrow \reals_{\ge 0}$ which is GSC but is not supermodular or ultra.    
\end{lemma}
\begin{proof}
    Consider the function $f:2^{[3]} \rightarrow \reals_{\ge 0}$ defined by $f(\{1\})=f(\{2\})=1, f(\{3\})=1$, $f(\{1,2\})=1.5,f(\{1,3\})=2, f(\{2,3\})=3, f(\{1,2,3\})=4$. 

To see that $f$ is not ultra we take $S=\{1\}, T=\{2,3\}$. Clearly no exchange leads to a non-decreasing sum of valuations.

To see that $f$ is not supermodular, we note that $f(1\mid \{2\}) < f(1\mid \emptyset)$.

We claim that $f$ is GSC with $S_1 = \{1,2\}, S_2=\{3\}$. The demand-cover vectors of this GSC partition are the standard basis and $\pm(1,-1, 0), \pm(1, 0, 1), \pm (0, 1,1)$. Thus, we need to show \[
V(f) \cap \{(1,1,1), (-1, 1, 1), (1,-1, 1), (1, 1, -1), (1, 1, 0), (1, 0, -1), (0, 1, -1)\} =\emptyset.
\]
We do this by a case analysis:
\begin{enumerate}
    \item $(1,1,1) \notin V(f)$: Assume towards contradiction that there exists a price vector $p$ such that $\emptyset$ and $\{1,2,3\}$ are demanded. Note that $\{1\}$ is also demanded.
    \item Let $i\in [3]$, $\begin{cases}
        1 & j\ne i \\
        -1 & j =i.
    \end{cases}\notin V(f)$: Assume towards contradiction that there exists a price vector $p$ such that $\{i\}$ and $[3]\setminus \{i\}$ are demanded. Clearly $[3]$ is also demanded.
    \item $(1,1,0)\notin V(f)$: $f|_{\{1,2\}}$ and $f(\circ\mid 3)$ are GS.
    \item $(1,0,-1)\notin V(f)$: $f|_{\{1,3\}}$ and $f(\circ\mid 2)$ are supermodular.
    \item $(0, 1, -1)\notin V(f)$: $f|_{\{2,3\}}$ and $f(\circ \mid 1)$ are supermodular.
\end{enumerate}

\end{proof}

\begin{lemma} [Supermodular $\setminus ($ Ultra $\cup$ GSC$)$]
    There exists a set function $f:2^{[3]}\rightarrow \reals_{\ge 0}$ which is supermodular, but is not ultra or GSC.
\end{lemma}
\begin{proof}
    Consider the function $f:2^{[3]}\rightarrow \reals_{\ge 0}$ defined by 
    \[
    f(S) = \begin{cases}
        |S|^2 & \text{if } S\ne \{2,3\} \\
        |S|^2 + 0.01 & \text{else.}
    \end{cases}
    \]
    Clearly $f$ is supermodular. To see that it is not ultra, we take $S=\{1\}, T=\{2,3\}$. Clearly no exchange leads to a non-decreasing sum of valuations.
    To see that $f$ is not GSC, we show that $\{(1,1,0), (1,0,1), (0, 1,1)\}\subseteq V(f)$. This rules out any partition which satisfies the GSC condition, since any two (of the three) elements must be in opposite parts. 
    Indeed, let $S\subseteq [3]$ be of size $2$ (i.e., $|S|=2$). Consider $p\in \reals^n$ defined by
    \[
    p_i = \begin{cases}
        \frac{f(S)}{2} & \text{if }i\in S\\
        10 &\text{else.}
    \end{cases}
    \]
    It holds that $\demandset{f}{p} = \{\emptyset, S\}$, which shows $\chi_S \in V(f)$ by \Cref{lem:normal_vec_price_vec}, concluding the proof.  
\end{proof}
\begin{lemma}[$($ GSC $\cap$ Supermodular$)$ $\setminus$ ultra]
    There exists a set function $f:2^{[3]}\rightarrow \reals_{\ge 0}$ which is GSC and supermodular but is not ultra.    
\end{lemma}
\begin{proof} 
Consider the function $f:\{1,2,3\} \rightarrow \reals_{\ge 0}$ defined by $f(\{1\})=f(\{2\})=f(\{3\})=1$, $f(\{1,2\})=f(\{1,3\})=2, f(\{2,3\})=3, f(\{1,2,3\})=4$. 

To see that $f$ is not ultra we take $S=\{1\}, T=\{2,3\}$. Clearly no exchange leads to a non-decreasing sum of valuations.

Note that marginals are non-decreasing and therefore $f$ is supermodular.

We claim that $f$ is GSC with $S_1 = \{1,2\}, S_2=\{3\}$. The demand-cover vectors of this GSC partition are the standard basis and $\pm(1,-1, 0), \pm(1, 0, 1), \pm (0, 1,1)$. Thus, we need to show \[
V(f) \cap \{(1,1,1), (-1, 1, 1), (1,-1, 1), (1, 1, -1), (1, 1, 0), (1, 0, -1), (0, 1, -1)\} =\emptyset.
\]
We do this by a case analysis: 
\begin{enumerate}
    \item $(1,1,1) \notin V(f):$ Assume towards contradiction that there exists a price vector $p$ such that $\emptyset$ and $\{1,2,3\}$ are demanded. Note that $\{1\}$ is also demanded.
    \item Let $i\in [3]$, $\begin{cases}
        1 & j\ne i \\
        -1 & j =i.
    \end{cases}\notin V(f)$: Assume towards contradiction that there exists a price vector $p$ such that $\{i\}$ and $[3]\setminus \{i\}$ are demanded. Clearly $[3]$ is also demanded.
    \item $(1,1,0)\notin V(f)$: $f|_{\{1,2\}}$ and $f(\circ\mid 3)$ are additive.
    \item $(1,0,-1)\notin V(f)$: $f|_{\{1,3\}}$ and $f(\circ\mid 2)$ are additive.
    \item $(0, 1, -1)\notin V(f)$: $f|_{\{2,3\}}$ and $f(\circ \mid 1)$ are supermodular.
\end{enumerate}
\end{proof}

\begin{lemma} [$($ Ultra $\cap$ Supermodular $)\setminus$ GSC]
    There exists a set function $f:2^{[3]}\rightarrow \reals_{\ge 0}$ which is ultra and supermodular, but is not GSC.
\end{lemma}
\begin{proof}
    Consider the set function $f:2^{[3]}\rightarrow \reals_{\ge 0}$ defined by $f(S)=|S|^2$.
    Because it is symmetric, it must be ultra (as was observed in \citet{ultra_valuations}). Additionally, it is clearly supermodular.
    To see that $f$ is not GSC, we show that $\{(1,1,0), (1,0,1), (0, 1,1)\}\subseteq V(f)$. This rules out any partition which satisfies the GSC condition, since any two (of the three) elements must be in opposite parts. 
    Indeed, let $S\subseteq [3]$ be of size $2$ (i.e., $|S|=2$). Consider $p\in \reals^n$ defined by
    \[
    p_i = \begin{cases}
        \frac{f(S)}{2} & \text{if }i\in S\\
        10 &\text{else.}
    \end{cases}
    \]
    It holds that $\demandset{f}{p} = \{\emptyset, S\}$, which shows $\chi_S \in V(f)$ by \Cref{lem:normal_vec_price_vec}, concluding the proof.  
\end{proof}

\begin{lemma}[$($Ultra $\cap$ Supermodular $ \cap$ GSC$)\setminus$ GS]
    There exists a set function $f:2^{[3]} \rightarrow \reals_{\ge 0}$ which is ultra, supermodular, and GSC, but is not GS
\end{lemma}
\begin{proof}
    Consider the set function $f:2^{[3]}\rightarrow \reals_{\ge 0}$ defined by 
    \[
    f(S) = |S\cap \{1,2\}| \cdot  \begin{cases}
        1 & \text{if } 3\notin S \\
        2 &\text{else}.
    \end{cases}
    \]
    Clearly, it is supermodular, and it is GSC with partition $\{1,2\}, \{3\}$. Since it's not submodular, it is not GS. It thus remains to show that $f$ is ultra. The only non-trivial sets $S,T$ to consider in the ultra definition are partitions of $[3]$ such that $|S|=1, |T|=2$. If $S\ne \{3\}$, observe that one of $\{1,2\}$ must be in $T\setminus S$, and exchanging this for the item in $T$ has no impact on each of the valuations, and thus no impact on their sum, as needed. If $S=\{3\}$, then $T=\{1,2\}$ and $f(T)+f(S)=2$. exchanging $3$ for either $1$ or $2$ yields $f(T)+f(S) =1\cdot 2+1=3$, as needed.
\end{proof}

\begin{lemma} [$($ Ultra $\cap$ GSC$) \setminus$ $($ Supermodular $\cup$ GS$)$, Section~9.2 in \citet{ultra_valuations}]
    There exists set function $f:2^{[n]}\rightarrow \reals_{\ge 0}$ which is ultra and GSC but is not supermodular or GS.
\end{lemma}

\subsection{GSC+ $\ne$ GSC}
In this section we show that the containment between GSC and GSC+(which does no appear in the venn diagram) is in fact, strict.
\begin{example}\label{ex:gsc+_not_gsc}
    Consider the function $f:2^{[3]}\rightarrow [0,1]$ defined as
    \[
    f(S) = \begin{cases}
        0 & \text{if }|S|=0 \\
        0 & \text{if }|S|=1 \\ 
        1 & \text{if }|S|=2 \\
        1 & \text{if }|S|=3.
    \end{cases}
    \]
    This function is GSC+. As $n=3$ we need only demonstrate that $(1,1,1)$ is not in $V(f)$.  But by~\cref{lem:normal_vec_price_vec} if $(1,1,1)\in V(f)$ then there must exist prices $p$ at which the demand set is precisely $\{\emptyset,\{1,2,3\}\}$.  So assume there exist prices $p$ at which $f(\emptyset)=f(\{1,2,3\})-(p_1+p_2+p_3)>f(S)-\sum_{i\in S}p_i$ for all $\emptyset\subsetneq S\subsetneq\{1,2,3\}$.  But then $p_1+p_2+p_3=1$.  If $p_1\leq 0$ then $f(\{1\})-p_1=0-p_1\geq 0=f(\emptyset)$, contradicting our assumption.  But if $p_1>0$ then $f(\{1,2,3\})-(p_1+p_2+p_3)=1-(p_1+p_2+p_3)<1-(p_2+p_3)=f(\{2,3\})-(p_2+p_3)$, again contradicting our assumption.
    
    Additionally, this function is \emph{not} GSC, since any two items display complementarities and so no partition as in~\Cref{def:uncommon-set-functions}(\ref{def:gsc}) exists.
\end{example}

\section{Divergence from the setting of \citet{demand_types}.}\label{app:demand_types}
As mentioned in \Cref{sec:demand_types_prelim}, our setting diverges from that of \citet{demand_types}, in that they  allow the domain of the set function $f$ to be any finite subset of $\mathbb{Z}^{[n]}$, while we only consider set functions $f$ with domain $2^{[n]}$. This set function setting was also considered in \citet{baldwinworkingpaper}, and we now note two differences from the setup of \citet{demand_types}, which result from the restriction of the function domains to the form $2^{[n]}$:
\begin{enumerate}
    \item In our 
    setting, the demand type defined by a demand cover $\demandcover$ is empty unless $\demandcover$ contains $\standard{i}$ for all $i\in[n]$. This is in contrast to \citet{demand_types}, where the demand type is always non-empty. (see \Cref{lem:standard_vecs_cover}).
    \item 
    In our 
    setting, different demand covers can lead to the same demand type (see \Cref{lem:redundant_vecs}), in contrast to the more general setting of \citet{demand_types} where this cannot happen. In our setting this can happen because some vector $v$ in the demand cover may be ``redundant'' i.e., there exists no function $f$ of that demand type with  $v\in V(f)$. 
    Since we're interested in the explicit vectors as well, we introduce notation and terminology for demand covers, which was not needed in \citet{demand_types} due to the exact correspondence between demand types and demand covers (see our definitions in \Cref{sec:demand_types_prelim}).
\end{enumerate}

\begin{lemma}\label{lem:standard_vecs_cover}
    Let $n\in \positivenats$ and $\demandcover \subseteq \{-1, 0, 1\}^n$ be a demand cover. $\demandtype(\demandcover) = \emptyset$ unless $\standard{i}\in \demandcover$ for every $i\in [n]$.
\end{lemma}
\begin{proof}
    Let $f:2^{[n]}\rightarrow \reals_{\ge 0}$ be a set function, and let $i\in [n]$. We claim that $\standard{i}\in V(f)$. Let $M=\max_{S\subseteq [n]} f(S) + 1$. Observe that under the price vector $p\in \reals^n$ defined by \[
    p_j=\begin{cases}
        M & j\ne i \\
        f(\{i\})-f(\emptyset) & j=i,\\
    \end{cases}\]
     the demand set is $\demandset{f}{p}=\{\emptyset, \{i\}\}$, and thus by \Cref{lem:normal_vec_price_vec}, $\standard{i}\in V(f)$.
\end{proof}
\begin{lemma}\label{lem:redundant_vecs}.
    There exist $n\in \positivenats$ and different demand covers $\demandcover, \demandcover'\subseteq \{-1,0,1\}^n$ such that $\demandtype(\demandcover)=\demandtype(\demandcover')\ne \emptyset$.
\end{lemma}
\begin{proof}
    We show an example of a demand cover with a ``redundant'' vector, i.e., a vector $v$ such that for any function $f$ of that demand type $v\notin V(f)$.
    
    It is easier to express the example using vectors rather than sets.  Let $n=3$ and $\demandcover=\pm\{\standard{1},\standard{2},\standard{3},\standard{1}+\standard{2}+\standard{3}\}$. Write $v=\standard{1}+\standard{2}+\standard{3}$. Suppose $v\in V(f)$ for some $f\in\demandtype(\demandcover)$.  Then $f$ uniquely demands a bundle $w\in \{0,1\}^3$ on one side of some facet $F$, and $w+v\in \{0,1\}^3$ on the other side, which is only consistent with $w=0$.  As the unique demand regions in which $0$ and $v$ are demanded are both unique, $F$ is the unique facet in the LIP for $f$ with normal $v$.

Next we show that $F$ must somewhere intersect with another facet of the LIP for $f$.  Observe that $f(0,0,0)=f(v)-(p_1+p_2+p_3)$ for any $p$ in the relative interior of $F$; fix some such $p$.  If $F$ meets no other facet then, for every $K>0$, also $p-K\standard{1}+K\standard{2}\in F$.  But $f(\standard{1})-(p_1-K)>f(0)$ for any $K>f(0)-f(\standard{1})+p_1$, so bundle $\standard{1}$ is preferred to $0$ (and thus also to $v$ at price $p-K\standard{1}+K \standard{2}$, we conclude that this price cannot be in $F$.  So $F$ must meet another facet, $F'$.

It follows that $F$ must be one of a collection $\mathcal{F}$ of facets of $\lip_f$, intersecting in a $(n-2)$-dimensional object $G$ in their boundaries (\citet{demand_types} Proposition 2.7).  Moreover, the collection $\mathcal{F}$ is ``balanced'' around $G$: the sum of their normals, appropriately oriented, is zero (\citet{demand_types} Theorem 2.14; note that all ``weights'' are 1 in the setting of this paper).

Since $F$ is the unique facet in the LIP for $f$ with normal $v$, in particular it is the only facet with this normal in $\mathcal{F}$.  So, without loss, we may assume there is a facet with normal $\standard{1}$ in $\mathcal{F}$.  Then $G$ is spanned by $\standard{2}-\standard{3}$.  Thus there cannot be any facet in $\mathcal{F}$ with normal either $\standard{2}$ or $\standard{3}$.  But it is now impossible for the collection $\mathcal{F}$ to be balanced.  This contradiction demonstrates the impossibility of a function $f\in\demandtype(\demandcover)$ satisfying $v\in\demandcover(f)$.
\end{proof}
\section{Demand Covers Appendix}\label{app:demand_covers}
In this appendix we formally prove exact/minimal demand covers for the known classes of functions. Usually, for exact demand covers the fact that the class is exactly equal to the demand type is already known (for example, \citet{sun_and_yang} provide an appropriate demand cover for GSC), but in our setting, to show exactness/minimality of a demand cover, we also need to show that no vectors are redundant (see \Cref{remark:redudant_vectors}). 
In \Cref{app:demand_cover_tools} we present two useful tools for proving a demand cover's exactness/minimality, and in \Cref{app:demand_covers_minimality} we prove exactness/minimality of the demand covers presented in \Cref{sec:demand-covers}.

\subsection{Toolbox}\label{app:demand_cover_tools}
Observe that \Cref{lem:axis_parallel_demand_sets} (a generalization of Theorem~1 by
\citet{baldwinworkingpaper}) can be restated as follows.
\begin{lemma}\label{lem:axis_parallel_equivalent}
    Let $\demandcover \subseteq \{-1, 0, 1\}^n$ be a demand cover. A set function $f: 2^{[n]}\rightarrow \reals_{\ge 0}$ belongs to $\demandtype(\demandcover)$ if and only if for any price vector $p\in \reals^n$, item $i \in [n]$ and $\delta > 0$, if $S\in \demandset{f}{p}$ then there exists $T\in \demandset{f}{p+\delta \standard{i}}$ such that either $T=S$ or $\chi_T-\chi_S \in \demandcover$. 
\end{lemma}
This statement fits better with some of our set function classes definitions,  and we therefore use it in some of the following proofs.

Another useful tool is the following lemma, which facilitates the identification of $V(f)$ given $f$. It shows that the vectors in $V(f)$ are exactly those giving the difference between demanded bundles in cases where exactly two bundles are demanded.

\begin{lemma} \label{lem:normal_vec_price_vec}
    For any set function $f:2^{[n]}\rightarrow \reals_{\ge 0}$ and any vector $v\in \{-1,0,1\}^n$, it holds that $v\in V(f)$ if and only if there exists a price vector $p\in \reals^n$ such that $\demandset{f}{p}=\{S,T\}$ where $S,T\subseteq [n]$ such that $\chi_S-\chi_T = v$.
\end{lemma}
\begin{proof}
The result follows from Proposition 2.20, ``duality'', in \citet{demand_types}.\footnote{To translate from the language of that result to the present setting, observe that facets are $(n-1)$-dimensional ``price complex cells'' (Definition 2.5 of \citet{demand_types}).  Next, the  ``demand complex cells'' (Definition 2.15 of \citet{demand_types}) are the sets $\mbox{conv}\{\chi_S:S\in \demandset{f}{p}\}$, which (in our setting of $2^{[n]}$ domain) are 1-dimensional if and only if $\demandset{f}{p}$ contains exactly two elements.}  This shows  (main part and additional parts 1 and 4), that there is a bijective correspondence between two-element demand sets $\demandset{f}{p}=\{S,T\}$, and facets; label the corresponding facet $F_{S,T}$.  Moreover, a price $p$ is in the relative interior of $F_{S,T}$ if and only if $\demandset{f}{p}=\{S,T\}$, and $(p'-p)\cdot(\chi_S-\chi_T)=0$ for all $p,p'$ in $F_{S,T}$.

But by definition, $v\in V(f)$ if and only if there exists a facet $F\subset \reals^n$ with $(p'-p)\cdot v=0$ for all $p,p'\in F$.  And because $F$ is $(n-1)$-dimensional, this defines $v$ up to sign, demonstrating that $v=\pm(\chi_S-\chi_T)$; we may now re-order $S$ and $T$ if necessary.
\end{proof}

\subsection{Proofs of Minimality/Exactness} \label{app:demand_covers_minimality}
In this section we show the minimality of the demand-cover families presented in \Cref{sec:demand-covers}. 

\subsubsection{Minimality of the Demand-Cover Family for Gross Substitutes} 
\begin{lemma} \label{lem:gs_cover_minimality}
    The demand-cover family from \Cref{prop:gross_covers}(\ref{item:gc_cover}) is minimal for the class of gross substitutes functions.
\end{lemma}

\begin{proof}
    Fix any $n\in \positivenats$, and let $v \in \{-1,0,1\}^n$ be a vector from the demand-cover family in \Cref{prop:gross_covers}(\ref{item:gc_cover}) (i.e., a vector with at most one positive entry and at most one negative entry). We need to show a gross substitutes function $f:2^{[n]}\rightarrow \reals_{\ge 0}$ such that $v\in V(f)$. 
    If $v$ has only one non-zero coordinate, then any function (and in particular any gross substitutes function) satisfies $v\in V(f)$, and we are done.

    Otherwise, it holds that $v=\standard{i}-\standard{j}$ for some $i\ne j\in [n]$. Consider the function $f:2^{[n]}\rightarrow \reals_{\ge 0}$ defined by 
    \[
    f(S) = \indicator[S\ne \emptyset].
    \]
    This function is unit demand\footnote{A set function $f:2^{[n]}\rightarrow \reals_{\ge 0}$ is unit demand if, for every $S\subseteq [n]$ it holds that 
$f(S) \;=\; \max_{i \in S} f(\{i\})$. This is a known strict-subclass of gross substitutes\citep{gul1999walrasian}.}, and therefore gross substitutes.
    Let $p\in \reals^n$ be the price vector defined as 
    \[
    p_k = \indicator[k\notin \{i,j\}] + 0.1
    \]
    It holds that $\demandset{f}{p}=\{\{i\}, \{j\}\}$, and therefore $v\in V(f)$ by \Cref{lem:normal_vec_price_vec}.
\end{proof}

\subsubsection{Minimality of the Demand-Cover Family for Gross Complements}
\begin{lemma}\label{lem:supermodular_cover_minimality}
    The demand-cover family from \Cref{prop:gross_covers}(\ref{item:gc_cover}) is minimal for the class of gross complements functions.
\end{lemma}

\begin{proof}
    Fix any $n\in \positivenats$, and let $v \in \{-1,0,1\}^n$ be a vector from the demand-cover family in \Cref{prop:gross_covers}(\ref{item:gc_cover}) (i.e., a vector whose non-zero coordinate entries are all of the same sign). Consider the function $f:2^{[n]}\rightarrow \reals_{\ge 0}$ defined by
    \[
    f(S) = |S|^2.
    \]
    $f$ is clearly supermodular. Denote by $S$ the set of non-zero coordinates of $v$, and consider the price vector 
    $p\in \reals_{\ge 0}$ defined by 
    \[
    p_i = \begin{cases}
        \frac{f(S)}{|S|} &\text{if } i\in S \\
        2n^2 & \text{else.}
    \end{cases}
    \]
    It holds that $\demandset{f}{p}=\{\emptyset, S\}$, and therefore $v\in V(f)$, as needed.
\end{proof}

\subsubsection{Minimality of the Demand-Cover Family for GSC+}

\begin{lemma}\label{lem:gsc+_cover_minimality}
    The demand-cover family from \Cref{prop:gsc+} is minimal for the class of GSC+ functions.
\end{lemma}

\begin{proof}
    Fix any $n\in \positivenats$, and let $v \in \{-1,0,1\}^n$ be a vector from the demand-cover family in \Cref{prop:gsc+} (i.e., a vector with at most $2$ non-zero entries).
    Unless $v$ contains $2$ positive entries, or $2$ negative entries, there exists a gross substitutes (and thus, GSC+ as well) function $f$ with $v\in V(f)$, as needed (from \Cref{lem:gs_cover_minimality}.
    
    Consider the function $f:2^{[n]}\rightarrow \reals_{\ge 0}$ defined by
    \[
    f(S) = \begin{cases}
        0 & \text{if }|S|<2\\
        1 & \text{if }|S|\ge 2.
    \end{cases}
    \]
    It is easy to see that $f$ is GSC+. Denote by $S$ the set of non-zero coordinates of $v$, and consider the price vector 
    $p\in \reals_{\ge 0}$ defined by 
    \[
    p_i = \begin{cases}
        \frac{f(S)}{|S|} &\text{if } i\in S \\
        2n^2 & \text{else.}
    \end{cases}
    \]
    It holds that $\demandset{f}{p}=\{\emptyset, S\}$, and therefore $v\in V(f)$, as needed.
\end{proof}
\subsubsection{Minimality of the Demand-Cover Family for $\Delta$-Substitutes}
\begin{lemma}\label{lem:delta_sub_cover_min}
    The demand-cover family from \Cref{prop:near-substitutes_demand_type} is minimal for the class of $\Delta$-substitutes.
\end{lemma}
\begin{proof}
    Fix any $\Delta, b\in \positivenats$, and let $v\in \{-1, 0,1\}^n$ be a vector from the demand-cover family in \Cref{prop:near-substitutes_demand_type} (i.e., a vector with at most $\Delta$ positive entries and at most $\Delta$ negative entries). \emph{Suppose $v$ only contains non-positive entries or non-negative entries.} Denote by $k$ the number of non-zero entries in $v$, and observe that $k\le \Delta$. Consider the set function $f:2^{[n]}\to \reals_{\ge 0}$ defined by $f(S) = \indicator[|S| \ge k]$. This functions is clearly $\Delta$-substitutes. We claim that $v\in N(f)$. Indeed, the price vector $p\in \reals^n$ defined by \[
    p_i = \begin{cases}
        1/k & v_i \ne 0 \\
        2 & \text{else},
    \end{cases}
    \]
    satisfies $\demandset{f}{p} = \{\emptyset, \{i\in [n]\mid v_i \ne 0\}\}$, and thus by \Cref{lem:normal_vec_price_vec}, $v\in N(f)$, as needed.
    
    Denote by $P=\{i\in [n]\mid v_i =1\}$ the set of positive coordinates, and by $N=\{i\in [n]\mid v_i = -1\}$ the set of negative coordinates. Consider the set function $f:2^{[n]}\to \reals_{\ge 0}$ defined by $f(S) = \indicator[P\subseteq S \lor N \subseteq S]$. Firstly, it is clearly $\Delta$-substitutes, since $|P|, |N|\le \Delta$. Secondly, consider the price vector $p\in \reals^n$ defined by 
    \[
    p_i = \begin{cases}
        \frac{1}{2|P|} & i\in P \\
        \frac{1}{2|N|} & i\in N \\
        2 & \text{else}.
    \end{cases}
    \]
    It holds that $\demandset{f}{p}=\{P,N\}$, and thus by \Cref{lem:normal_vec_price_vec}, $v\in N(f)$, as needed.
\end{proof}
\subsubsection{Proof of \Cref{prop:gsc+}}\label{app:gsc+}
\gscplus*
\begin{proof}
    The cover in \Cref{prop:gsc+} covers all GSC+ functions by \Cref{lem:axis_parallel_equivalent}, and all GSC functions by \citet{baldwinworkingpaper}, footnote 19.
    It remains to show that there are no redundant vectors, thus for each vector $v\in \{-1,0,1\}^n$ we show a GSC (and hence, GSC+) function $f:2^{[n]}\to \reals_{\ge 0}$ such that $v\in V(f)$. Since GS is contained in GSC+, and the cover from \Cref{prop:gross_covers}(\Cref{item:gs_cover}) is exact, we only need to handle the case of $v=\standard{i}+\standard{j}$ for some $i\ne j\in [n]$.
    Consider the function $f:2^{[n]}\to \reals_{\ge 0}$ defined by 
    \[
    f(S) = \indicator[\{i,j\} \subseteq S].
    \]
    It is clearly GSC (with partition $\{i\}$, $[n]\setminus \{i\}$), and at prices \[
    p_k = \begin{cases}
        1/2 & k=i,j \\
        2 &\text{else},
    \end{cases}
    \]
    it holds that $\demandset{f}{p}=\{\emptyset, \{i,j\}\}$, and thus, by \Cref{lem:normal_vec_price_vec}, $v\in V(f)$, as needed.
\end{proof}
\subsubsection{Proof of \Cref{prop:gooddemand_type}}\label{app:asc_exact}
The fact that ASC is exactly the demand-type family defined by the cover in \Cref{prop:gooddemand_type} is implied by \Cref{lem:axis_parallel_equivalent}. We need to show that the covers contain no redundant vectors. This is implied by our proof of \Cref{prop:ultra_demand}, where we first show that all ultra functions are ASC, and then show that the cover is minimal for ultra, and thus also for ASC.

\subsubsection{ASC is the Minimal Containing Demand-Type Family of Ultra}\label{sec:asc_min_ultra}
\ultrademandcover*
\begin{proof} [Proof of \Cref{prop:ultra_demand}]
    We need to show two things: (1) that \gooddemand\ contains ultra, and (2) that \gooddemand\ is minimal, subject to (1).
    
    \textit{(1) \gooddemand\ contains ultra:}
    Let $f:2^{[n]}\rightarrow \reals_{\ge 0}$ be an ultra function, and let $v\in V(f)$. We claim that $v$ satisfies one of the conditions in \Cref{prop:gooddemand_type}. 
    By \Cref{lem:normal_vec_price_vec}, there exists a price vector $p\in \reals_{\ge 0}^n$ such that $\demandset{f}{p} = \{S, T\}$ where $\chi_S-\chi_T=v$. Assume without loss of generality that $|S| \le |T|$ (since $v$ satisfies one of the \gooddemand\ conditions if and only if $-v$ does).
    If $S\subseteq T$, then $v$ is entirely non-positive, as needed.
    Otherwise, there exists some $x\in S\setminus T$. By applying the ultra condition there exists some $y\in T\setminus S$ such 
    \[
    f(S)+f(T) \le f((S\setminus \{x\})\cup \{y\})+f((T\setminus \{y\})\cup \{x\}).
    \]
    subtracting $\sum_{i\in S} p_i+\sum_{i\in T} p_i$ from both sides yields
    \[
    f(S) - \sum_{i\in S} p_i + f(T) -\sum_{i\in T} p_i \le f((S\setminus \{x\})\cup \{y\}) - \sum_{i\in (S\setminus \{x\})\cup \{y\}}p_i + f((T\setminus \{y\})\cup \{x\})-\sum_{i\in (T\setminus \{y\})\cup \{x\}}p_i.
    \]
    Because $\demandset{f}{p}=\{S,T\}$, this is only possible if $\{(T\setminus \{y\})\cup \{x\}, (S\setminus \{x\})\cup \{y\}\} \subseteq \{S,T\}$, which implies $T=(S\setminus \{x\})\cup \{y\}$ and $v=e_x-e_y$, implying that it satisfies the second condition of \Cref{prop:gooddemand_type}.
    
    \textit{(2) \gooddemand\ is minimal such that it contains ultra:}
    Let $n\in \positivenats$ and let $v\in \{-1, 0, 1\}^n$ be a vector satisfying one of the conditions in \Cref{prop:gooddemand_type}, we will show an ultra function $f$ such that $v\in V(f)$. 
    
    If $v$ satisfies condition (2) of \Cref{prop:gooddemand_type}, then it is also a vector of the demand cover of GS, which is exact (\Cref{prop:gross_covers}(\ref{item:gc_cover})). 
    Therefore there exists some GS function $f$ such that $v\in V(f)$. 
    Since ultra contains GS, $f$ is also ultra and we are done.
    
    Otherwise, if $v$ satisfies condition (1) of \Cref{prop:gooddemand_type}, 
    let $I=\{i\in [n]\mid v_i \ne 0\}$ be the set of non-zero indices of $v$. Consider the symmetric set function $f:2^{[n]}\rightarrow \reals_{\ge 0}$ defined by
    \[
    f(S) = \begin{cases}
        0 & |S| < |I| \\
        1 & |S| \ge |I|.
    \end{cases}
    \]
    To see that $v\in V(f)$, by \Cref{lem:normal_vec_price_vec}, it suffices to show a vector of prices $p$ such that $\demandset{f}{p} = \{I,\emptyset\}$. Indeed, the vector defined by
    \[
    p_i = \begin{cases}
        \frac{1}{|I|} & i\in I \\
        2 & i \notin I,
    \end{cases}
    \]
    satisfies this.
\end{proof}
\section{Deferred Proofs from \Cref{sec:poly_many_result}}\label{app:main_result}
\subsection{Proof of \Cref{lemma:facet_piercing_measure_0}}
\facetpiercingmeasure*
\begin{proof}
    Let $B$ be the set of vectors which are not facet-piercing with respect to $f$. 
    To see that $B$ is of measure $0$, we note that, for any $c\in \reals^n$, if $\ray{c}$ intersects with more than one facet there exist $\alpha >0$ and unique sets $S_1, S_2, S_3 \subseteq \actions$ such that the agent is indifferent between $S_1, S_2$ and $S_3$, under $\alpha$, i.e., 
    \[
    \alpha \cdot f(S_1)-\sum_{i\in S_1} c_i = \alpha \cdot f(S_2)-\sum_{i\in S_2} c_i = \alpha \cdot f(S_3)-\sum_{i\in S_3} c_i.
    \]
    For any fixed distinct $S_1, S_2, S_3$, we show the system of equalities imposes a non-trivial constraint on $c$, so its solution set has measure $0$. Subtracting pairs of equalities gives
    \begin{align*}
        \langle\chi_{S_1}-\chi_{S_2}, c \rangle &= \alpha\cdot(f(S_1)-f(S_2)), \\
        \langle \chi_{S_1}-\chi_{S_3}, c \rangle &= \alpha\cdot(f(S_1)-f(S_3)).
    \end{align*}
    Since $S_1\ne S_2$ and $S_1\ne S_3$, both $\chi_{S_1}-\chi_{S_2}$ and $\chi_{S_1}-\chi_{S_3}$ are non-zero vectors, so each equation is a non-degenerate linear constraint on $(c,\alpha)\in\reals^n\times\reals$. Since $S_2 \ne S_3$, together they define a co-dimension-$2$ affine subspace of $\reals^{n+1}$. Projecting onto the $c$-coordinates (by marginalizing over $\alpha>0$) yields a set of measure $0$ in $\reals^n$. Since there are finitely many choices for $S_1, S_2, S_3$, $B$ is a finite union of measure-$0$ sets, hence also of measure $0$, as needed.
\end{proof}
\subsection{Proof of \Cref{lemma:almost_all_costs_wlg}}
\almostallcostswlg*
\begin{proof}
    Following the lead of \citet{DuettingEFK21}, define an $\varepsilon$-perturbation of a cost vector $c\in \reals_{\ge 0}^n$, to be a cost function $\hat{c}$ that satisfies $|\hat{c}_i - c_i| < \varepsilon$ for every $i\in [n]$. \citet{DuettingEFK21} show that for a small enough $\varepsilon$, $\varepsilon$-perturbations of $c$ admit at least as many critical points as $c$ itself:
    \begin{lemma} [Lemma~4.13 in \citep{DuettingEFK21}] \label{lemma:perturbation_lemma}
        For every success probability function $f$ and a cost vector $c$, there exists an $\varepsilon > 0$ such that for every $\varepsilon$-perturbation $\hat{c}$ of $c$,  there are at least as many critical points under $f$ and $\hat{c}$ as there are under $f$ and $c$.
    \end{lemma}
    The set of all $\varepsilon$-perturbations, for a small enough $\varepsilon$ per Lemma~\ref{lemma:perturbation_lemma}, is of measure $>0$, which implies the existence of $\hat{c}\in U$ which is an $\varepsilon$-perturbation of $c$, concluding the proof.
\end{proof}

\subsection{Proof of \Cref{lem:cost_incr_with_alpha}}\label{app:cost_increasing}
\brcostmonotone*
\begin{proof}
    By definition we know $f(S_1)-\frac{c(S_1)}{\alpha_1}\ge f(S_2)-\frac{c(S_2)}{\alpha_1}$ and 
    $f(S_2)-\frac{c(S_2)}{\alpha_2}\ge f(S_1)-\frac{c(S_1)}{\alpha_2}$.  Adding these inequalities and canceling duplicate terms yields 
    $\frac{c(S_1)}{\alpha_1}+\frac{c(S_2)}{\alpha_2}\le \frac{c(S_2)}{\alpha_1}+\frac{c(S_1)}{\alpha_2}$, that is, $(\alpha_2-\alpha_1)c(S_1)\le (\alpha_2-\alpha_1)c(S_2)$.  Our assumption that $\alpha_2>\alpha_1$ now provides the result.
\end{proof}
\section{On the Maximality of ASC} \label{app:asc_max}
In this appendix, we discuss whether ASC is the maximal demand-type family which admits poly-many critical points.

We conjecture that it is, subject to the constraints of dimensional-invariance and anonymity:
\begin{definition}[Anonymous demand type]\label{def:anonymous}
    Let $\mathcal{C}$ be a demand type and $\demandcover_n$ be its exact demand cover.
    We say that $\mathcal{C}$ is \emph{anonymous}, if, for any $v\in \demandcover_n$ and permutation $\sigma:[n]\to [n]$, the vector $v_\sigma$ defined as $v_\sigma = (v_{\sigma(1)},\dots, v_{\sigma(n)})$ is also in $\demandcover_n$.
    We extend this definition to anonymous demand-type families in the natural way.
\end{definition}

\begin{definition} [Dimensionally-invariant demand-type family] \label{def:consistent_fam}
    Let $\mathcal{C}$ be a demand-type family with the exact demand-cover family $\demandcoverfam = \{\demandcover_n\}_{n\in \positivenats}$. We say that $\mathcal{C}$ is \emph{dimensionally-invariant} if, for any $v\in \demandcover_n$, the  vector $(v,0) \in \{-1,0,1\}^{n+1}$ is in $\demandcover_{n+1}$.
\end{definition} 
\begin{remark}\label{remark:consistency}
    Note that \Cref{def:demand_type_fam} is extremely broad, and doesn't induce any shared structure between dimensions. Therefore, some assumption that insures shared structure between dimensions is necessary for asymptotic claims. We view \Cref{def:consistent_fam} as a natural axioms to insure homogeneity of demand types within the family. 
\end{remark}
\begin{conjecture}\label{conj:asc_max}
Let $\mathcal{C}$ be an anonymous and dimensionally-invariant demand-type family which is not contained in ASC. The worst-case number of critical points when $f\in \mathcal{C}$ is super-polynomial.  
\end{conjecture}

We complement \Cref{remark:consistency} by remarking that some assumption of the form of anonymity is also necessary for \Cref{conj:asc_max}. For example, if we were to add a single vector with $\ell_1$ norm $\le M$ to the demand cover of ASC the number of critical points would be bounded by $O(2^M n^2)$, and thus, polynomial. This is because we can ``fix'' the agent's choice over the $\le M$ elements ``affected'' by that vector, and then the same bound over the number of critical points from \Cref{thm:gooddemand_poly_crit} applies.

The reason we conjecture \Cref{conj:asc_max} is because we can show this holds locally along the linear contract line. More formally,
\begin{lemma}\label{lem:superpoly_brs}
    Let $\mathcal{C}$ be an anonymous and dimensionally-invariant demand-type family which is not contained in ASC, with the exact demand-cover family $\demandcoverfam = \{\demandcover_n\}_{n\in \positivenats}$. There exist series of set functions $\{f_n\}_{n\in \positivenats}$ and cost vectors $\{c^n\}_{n\in \positivenats}$, such that for any $n\in \positivenats$, when considering  $f_n:2^{[n]}\rightarrow \reals_{\ge 0}$ and costs $c^n\in \reals_{\ge 0}^n$, it holds that $\ray{c^n}$ crosses $2^{\Omega(\sqrt{n})}$ facets, all with normal vectors which belong to $\demandcover_n$.
\end{lemma}

Note that this lemma doesn't imply the correctness of the conjecture, since our worst functions might have LIP facets with normal vectors not in $\demandcover_n$, which simply do not intersect with the linear contract ray. In our proof we present one of many possible such constructions. If there exists a construction such that all LIP facets have normal vectors in $\demandcover_n$, and not just along the linear contract ray, then \Cref{conj:asc_max} would follow.

The core of our construction for \Cref{lem:superpoly_brs} is the following sequence of sets:
\begin{lemma}\label{lem:base_sequence}
    For any $n\in \positivenats$, there exists a sequence of sets $S_1,\dots, S_k\subseteq [n]$ such that $k \ge \sqrt{2}^{\sqrt{n}}$, and for any $i\in [k-1]$, $S_i$ is yielded from $S_{i-1}$ in one of the following ways:
    \begin{enumerate}
        \item (Costly $1$-substitution) there exist $b<a\in [n]$ such that $S_{i} \setminus S_{i-1} = \{a\}$ and $S_{i-1} \setminus S_{i} = \{b\}$.
        \item (Costly $2$-substitution) there exists $a\in [n]$ and $b,c,d<a$ such that $S_{i} \setminus S_{i-1} = \{a,b\}$ and $S_{i-1} \setminus S_{i} = \{c,d\}$.
    \end{enumerate}
\end{lemma}
\begin{proof}
    Before defining $S_1, \dots, S_k$, we define $T_0,\dots, T_{m-1}$, such that $m\ge \sqrt{2}^{\sqrt{n}}$, and $c(T_0) < \dots < c(T_{m-1})$ but the transition property in \Cref{lem:base_sequence} doesn't necessarily hold. We then add intermediate sets in between the sets $T_i$, to achieve the sequence $S_1,\dots, S_k$ (which would obviously have $k \ge m$, as needed). 
    
    The idea for the construction of $T_0,\dots, T_{m-1}$ is to encode a ``binary counter'' over $\ell\le \sqrt{n}$ bits into the elements $[n]$. The $j$th least significant bit is encoded into the $\ell$ items $\{(j-1)\cdot \ell, \dots, j\cdot \ell -1\}$, such that it always holds that exactly one of these $\ell$ items is selected. More precisely, let $m=2^\ell$, for any $i\in \{0, \dots, m\}$, and $j\in [\ell]$ denote \[
    b_j(i) = \begin{cases}
        (j-1)\ell & \text{the $j$th least significant bit of $i$ is $0$} \\
         j\ell -1 & \text{else}.
    \end{cases}
    \]
    The set $T_i$ is defined by $T_i = \{b_j(i)\}_{j\in [\ell]}$.

    It remains to describe the intermediary sets we add in between the $T_i$s to satisfy the transition property. Observe that, for any $i\in [m-1]$, exactly one bit changes from $0$ in $i-1$ to $1$ (is ``turned on'') in $i$, call this bit $j^\star$. The idea is to increase by $1$ the element which corresponds to $j^\star$, and use this increase to ``turn off'' one of the $j^\star-1$ bits which are smaller than $j^\star$, in what is a costly $2$-substitution, after all bits which are smaller than $j^\star$ have been ``zeroed out'', we use a costly $1$ substitution to bring the element which corresponds to the $j^\star$ bit to its ``on'' position. More formally,  consider the set sequence $T_i^0,\dots, T_i^{j^\star}$ defined by
    \[
    T_{i-1}^k = \{b_j(i-1)\}_{ j\in [\ell]:j^\star< j} \cup \{b_{j^\star}(i-1)+k\} \cup \{b_{j}(i-1)\}_{j\in [\ell]: k\le j <j^\star} \cup \{b_{j}(i)\}_{j\in [\ell]: j\le k}.
    \]
    Observe that $T_{i-1}^0 = T_{i-1}$. Let $k\in [j^\star]$, it is easy to verify that the change between $T_{i-1}^{k-1}$ and $T_{i-1}^{k}$ is a costly $2$-substitution. Finally, the change between $T_{i-1}^{j^\star}$ and $T_{i}$ is a costly $1$-substitution, as needed.
\end{proof}
We are now ready to prove \Cref{lem:superpoly_brs}.
\begin{proof}[Proof of \Cref{lem:superpoly_brs}]
    Let $n_0\in \positivenats$ be such that there exists a function $f:2^{[n_0]}\rightarrow \reals$ such that $f\in \mathcal{C}$ but is not ASC.
    To show \Cref{lem:superpoly_brs}, for any $n> n_0$, we show a cost vector $c^n$ and a sequence $\emptyset=S_0, S_1,\dots, S_k \subseteq [n]$ such that $c^n(S_1) < \dots < c^n(S_k)$, for all $i\in [k]$ $\chi_{S_{i}} - \chi_{S_{i-1}} \in \demandcover_n$ and $k\ge \sqrt{2}^{\sqrt{n-n_0}}$. This suffices, since for any values $0<\alpha_1<\dots <\alpha_k<1$, we can find values $f(S_1), \dots, f_n(S_k)$ such that for any $i\in [k]$, $\alpha_i = \frac{c^n(S_{i}) - c^n(S_{i-1})}{f_n(S_i)-f_n(S_{i-1})}$. Then, the functions $f_n:2^{[n]}\rightarrow \reals_{\ge 0}$ defined by
    \[
    f_n(S) = \max_{i\in [k]: S_i \subseteq S} f_n(S_i),
    \]
    satisfy the conditions of \Cref{lem:superpoly_brs}, since for each $n$, the ray $\ray{c^n}$ moves through the UDRs of $S_0, S_1,\dots, S_k$, in order.

    Since there exists a function $f:2^{n_0}\rightarrow \reals_{\ge 0}$ which is in $\mathcal{C}$ but is not ASC, there exists a vector $v\in \demandcover_{n_0}$  with positive coordinates $P$ and negative coordinates $N$ such that $1\le |P|\le |N|$ and $2\le |N|$. This implies, by dimensional-invariance and anonymity, that for any $n> n_0$, a change that corresponds to adding $|P|$ items and removing $|N|$ items is covered by $\demandcover_n$.
    Let $n > n_0$.
    
    First, we define $c^n$ by $c_i^n = 2^i$. The sequence $S_1,\dots, S_k$ is derived from \Cref{lem:base_sequence}. Specifically, we use the items of $[n_0]$ as ``dummy'' items, to turn our $|P|$ for $|N|$ substitutions into costly $1$/$2$-substitutions over the items $[n]\setminus [n_0]$ per \Cref{lem:base_sequence}. To arrive from $S_0=\emptyset$ to the first set in the resulting sequence $T$, we add intermediary sets to the beginning of the sequence, where the items of $T$ are added one at a time.
    To enable our use of the items of $[n_0]$, we ensure that, for every $S_i$, $|S_i \cap [n_0]| \le \frac{n_0}{2}$, and proceed as follows.
    
    If $|P|\ge 2$, we can implement the sequence in \Cref{lem:base_sequence} over the items $[n]\setminus [n_0]$ by adding $|P|$ items and removing $|N|$ from $[n]$ in a straightforward manner. Firstly for a $1$-substitution over $[n]\setminus [n_0]$, we also add $|P|-1$ and remove $|N|-1$ of the items in $[n_0]$ (recall that $|N|\ge |P|$, so to ensure that this is possible we can always add more intermediate sets, where at each step a single item in $[n_0]$ is add). For a $2$-substitution we add $|P|-2$ and remove $|N|-2$.

    If $|P|=1$, then $|N|>|P|$. Any $1$-substitution over $[n]\setminus [n_0]$ can still be implemented by adding $|P|-1$ items of $[n_0]$ and removing $|N|-1$. A $2$-substitution is split into two parts: First, for items $a,b,c,d$ per \Cref{lem:base_sequence}, we first add $a$, remove $c,d$, and add $|P|-1$ and remove $|N|-2$ items from $[n_0]$. Second, we simply add $b$.

    Following the above strategy, we arrive at a sequence $\emptyset = S_0,S_1,\dots, S_k$ such that $k \ge \sqrt{2}^{\sqrt{n-n_0}/2}$ and for each $i\in [k-1]$, $\chi_{S_{i+1}}-\chi_{S_i}\in \demandcover_n$. It remains to show that costs are increasing. Indeed, at every change from $S_i$ to $S_{i+1}$, we are either adding a single item (which trivially increases costs), or adding and removing items, where the largest item we add is larger than all removed items. By our choice of costs, this yields an increase in costs as well, concluding the proof.
\end{proof}
\section{Limits of Demand Types}\label{app:limits}
In this appendix we show that the demand types machinery is not useful for handling \emph{coverage} or \emph{budget-additive} functions.
\begin{definition}[Coverage and budget-additive] A set function $f:2^{[n]} \rightarrow \reals_{\ge 0}$ is:
    \item \textit{Coverage} if there is a set of elements $U$, with associated positive weights  $\{w_u\}_{u\in U}$, and a mapping $h:[n] \rightarrow 2^U$ such that for every $S\subseteq [n]$, $f(S)=\sum_{u\in U } w_u \cdot \indicator[\exists i \in S.~u \in  h(i)]$.
    \item \emph{Budget additive} if there are $w_1,\ldots,w_n,B \in \mathbb{R}_{\geq 0}$ such that $f(S)=\min(B,\sum_{j\in S}w_j)$ for any set $S \subseteq [n]$.
\end{definition}
Coverage and budget-additive are incomparable subclasses of submodular.

\subsection{Coverage Functions}
The following theorem shows that the minimal containing demand type of coverage functions is extremely broad, thus ruling out any useful analysis of coverage functions via demand types.
\begin{theorem} \label{thm:coverage_demand_type}
    For any non-zero vector $v\in \{-1, 0, 1\}^n$, with at least one positive and at least one negative coordinate, there exists a coverage function $f:2^{[n]}\rightarrow \reals_{\ge 0}$ such that $v\in V(f)$.
\end{theorem}
\begin{proof}
    Let $v\in \{-1, 0, 1\}^n$ be a non-zero vector. Denote by $P\subseteq [n]$ the set of indices $i$ such that $v_i = 1$ and by $N\subseteq [n]$ the set of indices such that $v_i = -1$. 
    Observe that $P$ and $N$ are two disjoint sets of actions. 
    Recall that a coverage function $f(S)=\sum_{u\in U } w_u \cdot \indicator[\exists i \in S.~u \in  h(i)]$ for $S \subseteq A = [n]$ is given by a set of elements $U$ with associated weights $\{w_u\}_{u\in U}$ and a mapping $h: A\rightarrow 2^U$ from actions $i \in A = [n]$ to sets $h(i) \in 2^U$.
    
    Consider the coverage function $f$ given by $U=P\times N, w_u = \frac{1}{|U|}$ and 
    \[
    h(i) = \begin{cases}
        \{i\} \times N & i\in P \\
        P \times \{i\} & i\in N \\
        \emptyset & \text{else.}
    \end{cases}
    \]
    Then
    \[
    f(S) = \frac{1}{|U|} \left|\bigcup_{i\in S} h(i)\right|.
    \]
    \Cref{fig:coverage construction} shows an example of this construction with $P = \{p_1, p_2, p_3, p_4\}$ and $N=\{n_1, n_2, n_3\}$. The elements of $U$ are arranged in a grid, where red and blue rectangles represent the sets associated (via $h$) with the actions of $P$ and $N$ (respectively). As the expression above dictates, the value of $f(S)$ is the fraction of circles within the union of the sets of elements associated by $h$ with the actions in $S$. Thus, $f(P) = f(N) = 1$, as the union of all columns or rows covers all elements. 
    Note that $f(\{p_1, n_1\})=\frac{6}{12} < \frac{1}{4} + \frac{1}{3} =  f(\{p_1\}) + f(\{n_1\})$.
    A key property of the construction, which we will exploit below, is the following: as long as only columns (resp.~rows) are selected, the contribution of any additional column is additive, while if at least one row (resp.~column) is selected, the contribution of adding a column (resp.~row) diminishes by at least $\frac{1}{|U|}$.
    
    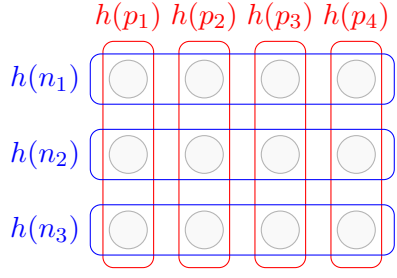
\begin{figure}
        \centering
         \begin{tikzpicture}
         \foreach \x in {1,...,4}{
            \foreach \y in {1,...,3} {
            \filldraw[color=gray!60, fill=gray!5](\x,\y) circle (1/4);
            }
         }
         \foreach \x in {1,...,4} {
             \draw[rounded corners, color=red] ({\x-1/3}, 1/2) rectangle ({\x+1/3}, 1/2+3) {};
             \node[anchor=south, color=red] at  (\x, 1/2+3) {$h(p_\x)$};
         }
         \foreach \y in {1,...,3} {
             \draw[rounded corners, color=blue] (1/2, {\y-1/3}) rectangle (1/2+4,{\y+1/3}) {};
             \node[anchor=east, color=blue] at  (1/2, 4-\y) {$h(n_{\y})$};
         }

        \end{tikzpicture}
        \caption{An illustration of the coverage function construction, with $P=\{p_1, p_2, p_3, p_4\}$ and $N=\{n_1, n_2, n_3\}$. The circles represent the elements of $U$. The red (respectively, blue) rectangles represent the sets associated with 
        actions $i \in P$ (resp., actions $i \in N$) by $h$. 
        For any action $i \not\in (P \cup N)$, $h(i) = \emptyset$.
        The value $f(S)$  is the fraction of circles within the union of the sets of elements associated with the actions in $S$.}\label{fig:coverage construction}
    \end{figure}

    Consider the price vector $p$ defined by \[
    p_i = \begin{cases}
        f(\{i\})-\frac{1}{2\cdot |N|\cdot |U|} & i\in N \\
        f(\{i\})-\frac{1}{2\cdot |P|\cdot |U|} & i\in P \\
        2 & \text{else}.
    \end{cases}
    \]
    We claim that $\demandset{f}{p}=\{N,P\}$, which concludes the proof by \Cref{lem:normal_vec_price_vec}.
    Let $S\in D_f(p)$, we claim that $S\in \{N,P\}$. This, combined with the fact that the utility from the sets $N$ and $P$ is $\frac{1}{2|U|}$, this shows $\demandset{f}{p} = \{S,T\}$, as needed. Observe that since the prices for items not in $N\cup P$ are greater than the maximal value of $f$, it must be that $S\subseteq N\cup P$. Assume towards contradiction that $S\not \subseteq P$ and $S\not \subseteq N$, this implies the existence of $i_P \in S\cap P$ and $i_N\in S\cap N$. But then, $S\setminus \{i_P\}$ has a strictly better utility:
    \[
    \begin{aligned}
    u_B(p, S) - u_B(p, S\setminus \{i_P\}) &= f(i_P\mid S\setminus \{i_P\}) - p_{i_P} \le f(i_P\mid \{i_N\}) - p_{i_P} \\
    &= f(\{i_P\}) - \frac{1}{|U|} - \left(f(\{i_P\}) - \frac{1}{2\cdot |P|\cdot |U|}\right) < 0,
    \end{aligned}
    \]
    where the second inequality is by submodularity. This contradicts $S\in \demandset{f}{p}$.
    
    Assume towards contradiction that $S\subsetneq N$ (similarly, if $S\subsetneq P$), then there exists $i\in N\setminus S$, but then $S\cup \{i\}$ has a strictly better utility:
    \[
    u_B(p, S\cup \{i\}) - u_B(p, S) = f(i\mid S) - p_i = f(\{i\}) - \left(f(\{i\}) - \frac{1}{2\cdot |N|\cdot |U|}\right) >0,
    \]
   contradicting $S\in \demandset{f}{p}$.

   All in all, $S\subseteq N\cup P$, it must be contained in one of them, and it cannot be strictly contained in either, and thus $S\in \{N,P\}$, as needed.
\end{proof}

\subsection{Budget-Additive Functions} 

We next show an analogous observation to Theorem~\ref{thm:coverage_demand_type}, but for budget-additive $f$. This again suggests that there might not be a meaningful demand type characterization for this class of functions.

\begin{theorem}
    For any distinct prime numbers $q,r$ and disjoint subsets $Q,R\subseteq [n]$ of sizes $q$ and $r$ respectively, there exists a budget-additive success probability function $f$ with a facet which is normal to the vector $v$ defined by
    \[
    v_i = \begin{cases}
        1 & i\in Q \\
        -1 & i\in R \\
        0 & \text{else}.
    \end{cases}
    \]
\end{theorem}

\begin{proof}
    Consider the budget additive function where 
    \[
    w_i = \begin{cases}
        \frac{1}{q} & i \in Q \\
        \frac{1}{r} & i \in R \\
        0 & \text{else},
    \end{cases}
    \]
    with budget $B=1$.
    Consider the price vector $p$ defined by 
    \[
    p_i = \begin{cases}
        (1-\frac{1}{2\cdot q \cdot r}) w_i & i\in Q\cup R\\
        2 & \text{else}.
    \end{cases}
    \]
    
    We start by claiming that $\demandset{f}{p} \subseteq \{Q,R\}$. Since $p_i = 2$ for any $i\notin Q \cup R$, clearly $S\subseteq Q\cup R$. Further observe that since prices are strictly greater than $w_i-1/qr$, going over budget is never profitable. 
    Furthermore, any set $T$ which does not exceed the budget and has elements from both $Q$ and $R$ must have $f(T) \le 1-1/qr$, and thus $u_B(p, T) = f(T) - p(T) = \frac{f(T)}{2qr} \le \frac{1-1/(qr)}{2qr}< \frac{1}{2qr}= u_B(p,Q)$. Thus, $S\subseteq Q$ or $S\subseteq R$. Any set $T$ which is strictly contained in $Q$ (or similarly, $R$), has $f(T) \le 1-\frac{1}{q}$, and thus $u_B(p, T) = f(T) - p(T) = \frac{f(T)}{2qr}\le \frac{1-1/q}{2qr} < u_B(p, Q)$. So it must be that $\demandset{f}{p}$ only contains items which are contained, but not strictly contained in $Q$ or $R$, and thus $\demandset{f}{p}\subseteq \{Q,R\}$, as needed.
    Finally, observe that the utility from both $Q$ and $R$ under $p$ is $u_B(p, Q) = u_B(p, R) = \frac{1}{2\cdot q\cdot r}$, and thus $\demandset{f}{p} = \{Q,R\}$. This concludes the proof by \Cref{lem:normal_vec_price_vec}.\end{proof}
\end{document}